\documentclass[preprint, 10pt, review]{elsarticle}
\usepackage[usenames,dvipsnames]{xcolor}

\usepackage{amsmath,amssymb,amsfonts,graphicx,bbold,amsthm,mathtools}
\usepackage{setspace}
\usepackage{enumitem}
\usepackage{subfigure,epstopdf,algorithm}
\usepackage{array}
\usepackage{xcolor}
\usepackage{thmtools}

\declaretheorem[numbered=no]{definition}
\declaretheorem{proposition}

\usepackage{times}
\usepackage{eqparbox}
\usepackage[noend]{algpseudocode}

\biboptions{sort&compress}
\journal{Signal Processing}
\date{}

\begin{document}

%
%

\begin{frontmatter}
	\title{ On Hilbert Transform, Analytic Signal, and Modulation Analysis for Signals over Graphs}
	
	\author{Arun Venkitaraman}
	\ead{arunv@kth.se}
	
	\author{Saikat Chatterjee}
	\ead{sach@kth.se}
	\date{}
	\author{Peter H{\"a}ndel}
	\ead{ph@kth.se}
	\address{Department of Information Science and Engineering,\\
		School of Electrical Engineering and ACCESS Linnaeus Center\\            
		KTH Royal Institute of Technology,  
		SE-100 44 Stockholm, Sweden  .}
	%
	%
	%
	%
	%
	%
	%
\begin{abstract}
We propose Hilbert transform and analytic signal construction for signals over graphs. This is motivated by the popularity of Hilbert transform, analytic signal, and modulation analysis in conventional signal processing, and the observation that complementary insight is often obtained by viewing conventional signals in the graph setting. 
Our definitions of Hilbert transform and analytic signal use a conjugate-symmetry-like property exhibited by the graph Fourier transform (GFT), resulting in a 'one-sided'  spectrum for the graph analytic signal. The resulting graph Hilbert transform is shown to possess many interesting mathematical properties and also exhibit the ability to highlight anomalies/discontinuities in the graph signal and the nodes across which signal discontinuities occur.
Using the graph analytic signal, we further define amplitude, phase, and frequency modulations for a graph signal. 
We illustrate the proposed concepts by showing applications to synthesized and real-world signals. For example, we show that the graph Hilbert transform can indicate presence of anomalies and that graph analytic signal, and associated amplitude and frequency modulations reveal complementary information in speech signals.
\end{abstract}

\begin{keyword}
Graph signal, analytic signal, Hilbert transform, demodulation, anomaly detection.
 \end{keyword}
\end{frontmatter}

\section{Introduction}
\label{intro}
The analysis of data over networks or graphs poses unique challenges to the signal processing community, since data must be seen with due regard to the connections between  various data points or nodes of a graph \cite{Newman,Sandry3,Shuman}. Given the wealth of techniques and models in conventional signal analysis, it is desirable to extend existing concepts to signals over graphs\cite{pp_graph2,deeplearninggraphwavelets,multiscalegraphs}. 
The collective efforts along this line of thought have led to the emergence of the notion of signal processing over graphs \cite{Shuman,Sandry1,Sandry2,Sandry3}. 
%
In this paper, we generalize the concepts of Hilbert transform, analytic signal, and modulation analysis to signals over graphs. This is motivated by two observations. Firstly, Hilbert transform, analytic signal, and associated modulation analysis have been used extensively for one-dimensional (1D) and two-dimensional (2D) signals in various applications\cite{Cohen,Cusmariu,Zayed,Guanlei,Sarkar,Phi_AS,monogenic,2D_AS_1,Bernstein2014}. By extending modulation analysis to graphs, we endeavour to provide similar tools for signals over graphs. Secondly, viewing of 1D /2D signals in a graph setting has been shown to give additional insight into the signals, leading to improved performance in tasks such as compression and denoising\cite{Shuman,ArunEusipco15}. 
\begin{figure}[]
	\vspace{-.in}
	{\centering
		\subfigure[\hspace{-.0in}]{
			\includegraphics[width=1.8in]{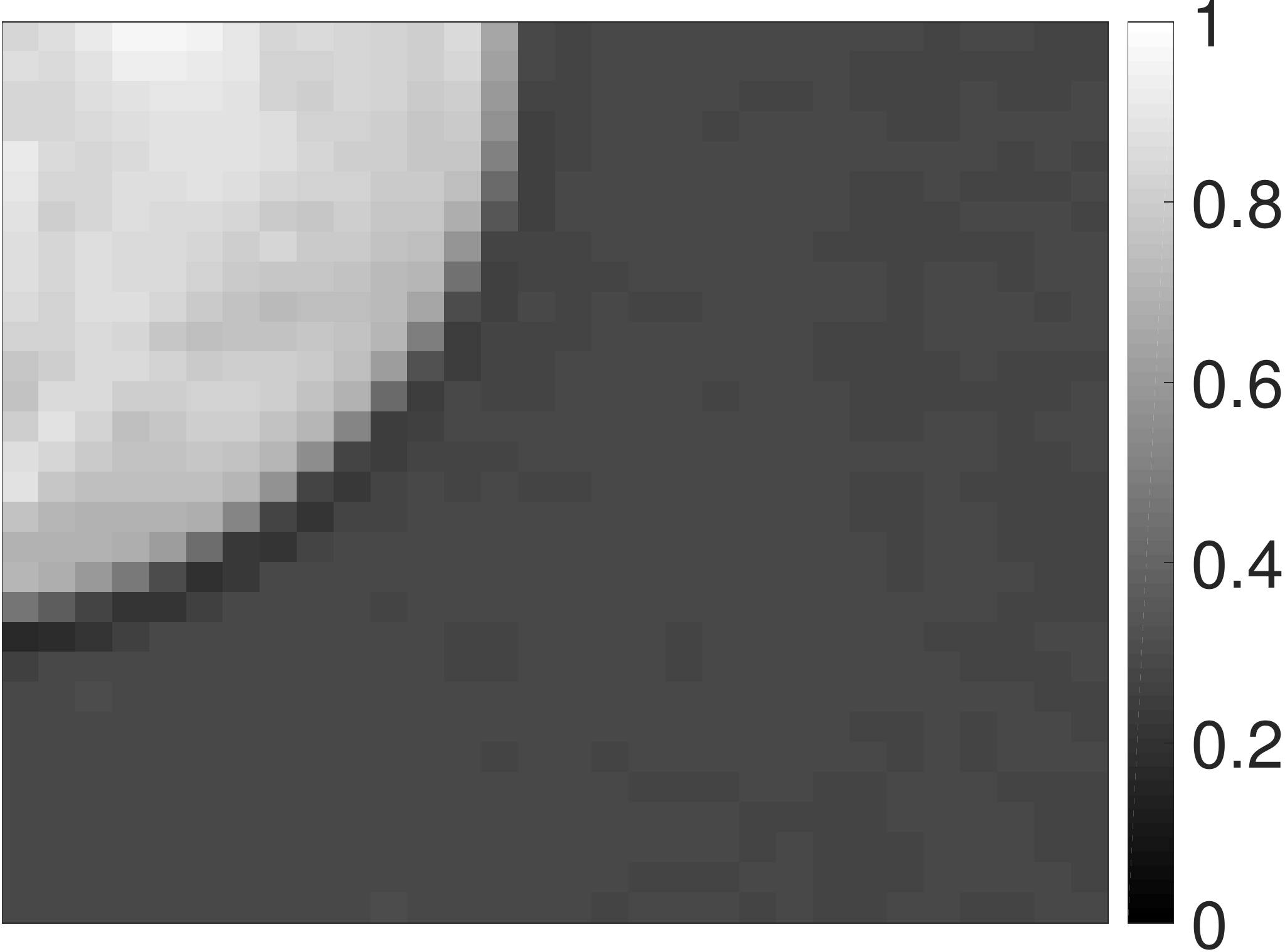}}
		\hspace{-.in}\subfigure[\hspace{-.0in}]{
			\includegraphics[width=1.8in]{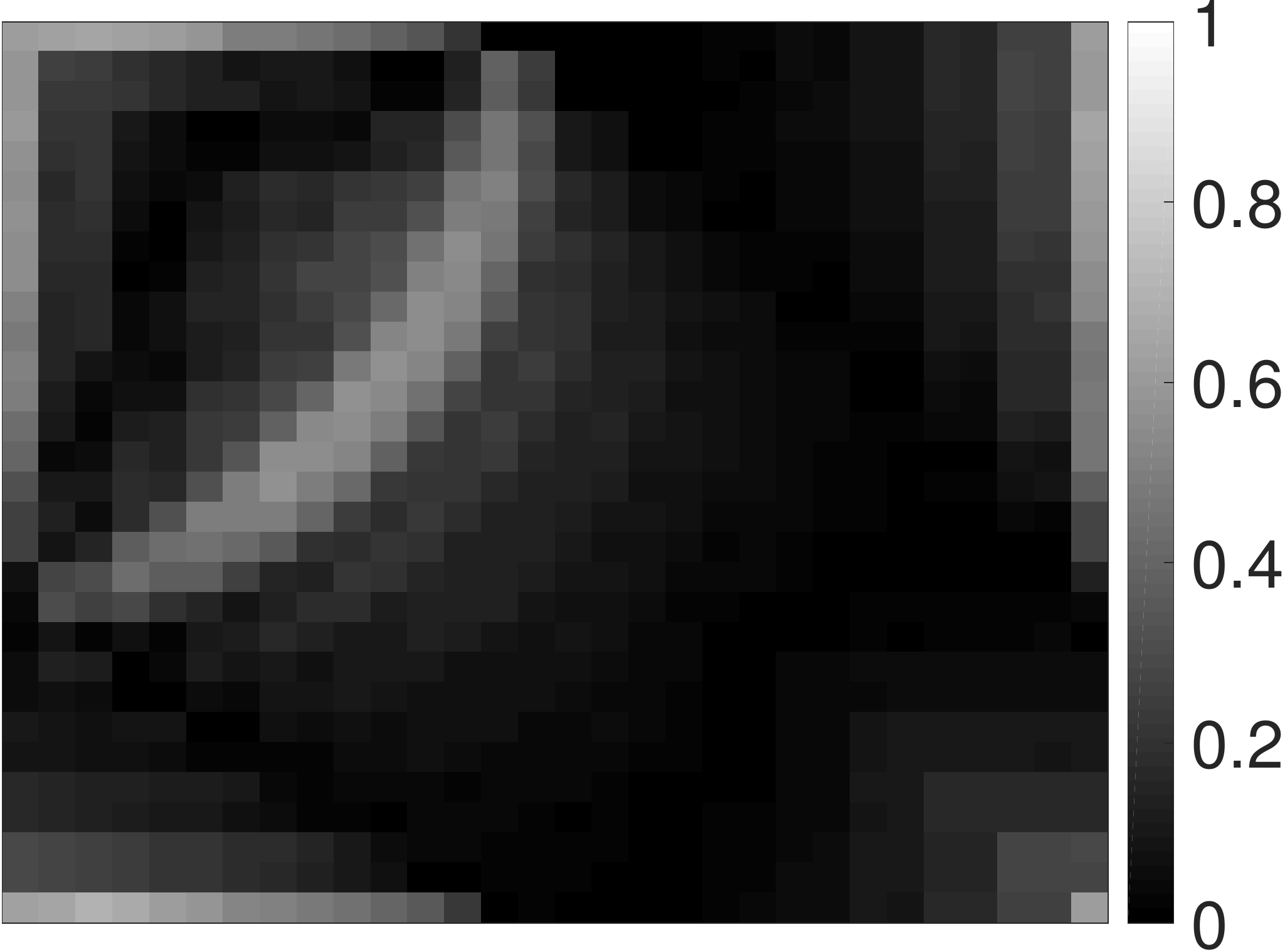}}
		\caption{Anomaly highlighting behavior of the graph Hilbert transform for 2D image signal graph. (a) image, and (b) graph Hilbert transform. We observe that the graph Hilbert transform highlights edges or sudden jumps across connected pixels.}
		\label{2Dedge_begin}
	}
\end{figure}
\subsection{Review of literature}
Some of the early works in graph signal processing include windowed Fourier transforms \cite{windowedGFT},  filterbanks\cite{ Narang2013,TAY201766}, wavelet transforms and multiresolution representations for graphs\cite{Coifman2006, Ganesan,Hammond2011,Wagner2,parseval_wavelets_gsp}. A number of strategies for efficient sampling of signals over graphs have been proposed \cite{graphsamp1,graphsamp3,graphsamp6,anis,chensamp,WANG2016119,PUY2016}. The notions of stationarity and power spectral density have also been considered extensively for signals over graphs \cite{statgraph1,statgraph2,statgraph4}.
 %
A parametric dictionary learning approach for graph signals was proposed by Thanou et al.\cite{Thanou2014}. In \cite{vertexfrequency}, Shuman et al. generalized the notion of time-frequency analysis to the graph setting using windowed graph Fourier transforms. Shahid et al. proposed variants of principal component analysis for graph signals  and developed scalable and efficient algorithms for recovery of low-rank matrices\cite{graphPCA1,graphPCA2}. Tremblay et al. proposed an efficient spectral clustering algorithm based on graph signal filtering\cite{Tremblay}. 
In \cite{benzisong}, Benzi et al. developed a song recommendation system based non-negative matrix factorization and graph total variation. Shuman et al. proposed a multi-scale pyramid transform on graphs that generates a multiresolution of both the graph and the signal\cite{pyramidgraph}. Segarra et al. proposed  convex optimization based approaches for blind identification of graph filters \cite{blinddeconvgraph1,blinddeconvgraph2}. Chen et al. considered signal recovery on graphs based on total-variation minimization formulated as a convex optimization problem\cite{chen1}. Sakiyama et al. proposed spectral graph wavelets and filterbanks constructed as sum of sinusoids in spectral domain with low approximation error\cite{sakiyama}. Deutsch et al. showed the application of spectral graph wavelets to manifold denoising\cite{deutsch}. A trilateral filter based denoising scheme was proposed by Onuki et al. \cite{onuki}. Multirate signal processing concepts including $M$-channel filter banks were extended to the graph setting by Teke and Vaidyanathan\cite{pp_graph2}. Kernel regression approaches for reconstruction of graph signals have also been recently proposed in the framework of reproducing kernel Hilbert spaces \cite{kergraph1}. Mendes et al. proposed a general framework for transforms for graph signals where they considered extension of tomograms to the graph signal setting \cite{tomograms_gsp}. A fast algorithm for implementation of vertex-frequency representations for graphs was developed by Jestrovic et al  in \cite{JESTROVIC2017483}. They also  developed an optimized vertex-frequency representation for investigating brain characteristics during consecutive swallows\cite{JESTROVIC2017113}. Kotzagiannidis and Dragotti extended the notion of finite-rate-of-innovations to circulant graphs \cite{fri_gsp} and also considered splines and wavelets for circulant graphs \cite{splines_gsp}.
 
  \subsection{Our contributions}
  \footnote{The codes related to our article may be found at https://www.kth.se/ise/research/reproducibleresearch-1.433797.}
In this paper, we propose definitions for the Hilbert transform and analytic signal for real signals over graphs\footnote{Part of this work has appeared in the Proceedings of the Sampling Theory and Applications Conference, 2015 \cite{ArunSampta15}.}.
We show that a real graph signal with a real-valued adjacency matrix may be represented using fewer number of GFT coefficients than the signal length, akin to the `one-sided' spectrum for 1D signals. We generalize the Hilbert transform and analytic signal construction \cite{Gabor} to graph signals by using the conjugate-symmetry-like property of the GFT basis. 
We also show that graph Hilbert transform and graph analytic signal inherit properties such as isometry, phase-shifting, and orthogonality from their 1D counterparts. 
 We discuss how the graph Hilbert transform does not possess a Bedrosian-type property in general unlike its conventional counterpart.
 As a natural consequence of the graph Hilbert transform construction, we propose amplitude, phase, and frequency modulations for graph signals.
We illustrate the concepts with applications to synthesized and real-world signals. 
Our experiments show that graph Hilbert transform can reveal edge connections and presence of anomalies in many graphs of interest. 
As an applicaton of the graph amplitude and frequency modulations, we also demonstrate that viewing the speech signal as a graph signal brings improves speaker classification performance. We summarize the key similarities and differences between the conventional Hilbert transform/analytic signal and our graph Hilbert transform/graph analytic signal in Table \ref{summary}. 

\begin{table}
	\centering
	\begin{tabular}{|c|c|c|c|}
		\hline
		Hilbert transform/analytic signal&  Graph Hilbert transform/analytic signal\\
		\hline
		Linear and shift invariant & Linear and graph-shift-invariant \\ \hline
		Highlights signal edges& Highlights  graph anomalies\\ \hline
		Quadrature phase-shifting& Generalized\\
		& quadrature phase-shifting \\\hline
		Changes with permutation of samples& Invariant to  node permutations\\\hline
		Bedrosian property & Lacks Bedrosian property \\\hline
		One-sided spectrum & Near one-sided spectrum\\  \hline
		Special case of graph Hilbert  & Added insight into 1D signals
		\\ 
		transform/graph analytic signal&\\
		\hline
\end{tabular}
	\hspace{.1in}
	\caption{Key similarities and differences between conventional Hilbert transform/ analytic signal and proposed graph Hilbert transform/analytic signal.}
	\label{summary}
	\vspace{-.in}
\end{table}
\vspace{-0.in}
\section{Preliminaries}
\subsection{Graph signal processing}
Let ${\bf x}\in\mathbb{R}^N$ be a real signal on the graph $G=(\mathcal{V},{\bf A})$, where $\mathcal{V}$ and $\mathbf{A}\in\mathbb{R}^{N\times N}$ denote the node set and the adjacency matrix, respectively. 
Then, the GFT of $\mathbf{x}$ is defined as\cite{Sandry1,Sandry2}:
\begin{equation}
\hat{\mathbf{ x}} \triangleq [\hat{x}(1), \hat{x}(2) , \hdots , \hat{x}(i), \hdots , \hat{x}(N)]^{\top}= \mathbf{V}^{-1}\mathbf{x},\nonumber
\end{equation}
where $\mathbf{V}$ denotes the matrix of generalized eigenvectors as columns (hereafter refered to simply as the eigenvectors) such that the Jordan decomposition of $\mathbf{A}$ is given by
$\mathbf{A=VJV}^{-1}$,
 and $\mathbf{J}$ is the matrix of Jordan blocks of the eigenvalues of $\mathbf{A}$. In the case when $\mathbf{A}$ is diagonalizable, $\mathbf{J}$ becomes the diagonal eigenvalue matrix $\mathbf{J}=\mbox{diag}(\lambda_1,\lambda_2,\cdots,\lambda_N)$. A periodic 1D signal may be viewed as a graph signal $\mathbf{x}$ with adjacency matrix 
\begin{equation}
\mathbf{A}=\mathbf{C} \triangleq \small\scriptstyle\left(\begin{array}{cccccc}
0& 1& 0&\cdots& 0\\
0&0&1&\cdots&0\\
\vdots&\vdots&\vdots &\vdots &\\
1&0&0&\cdots&0
\end{array}\right),\nonumber
\end{equation}
 and the GFT coincides with the discrete Fourier transform (DFT)\cite{Sandry3,gray_circulant}. 
 
 The smoothness of a graph signal is often measured in terms of the following mean-squared cost:
$\mbox{MS}_g(\mathbf{x})\displaystyle=\frac{\|\mathbf{x-Ax}\|_2^2}{\|\mathbf{x}\|_2^2}$.
A graph signal with low $\mbox{MS}_g$ is smooth over the graph: connected nodes have signal values close to each other, stronger the edge, closer the values. 
A unit shift of $\mathbf{x}$ over the graph is defined as $\mathbf{Ax}$, generalizing unit delay for the 1D graphs. 

A linear shift-invariant filter $\mathbf{H}$ on graph is defined as a polynomial $h(\cdot)$ of the adjacency matrix, such that
$\mathbf{H}=\sum_{l=0}^Lh_l\mathbf{A}^l=h(\mathbf{A})$,
where $h_l\in\mathbb{R}$ and $L\leq N$. Such a filter also follows the convolution property: the GFT coefficients of filtered graph signal 
$\mathbf{y}=\mathbf{H}\mathbf{x}$ are obtained by scaling the GFT coefficients of  input $\mathbf{x}$: $\hat{\mathbf{y}}=h(\mathbf{J})\hat{\mathbf{x}}$.
Graph filters are used in spectral analysis and processing of graph signals and have been applied in various applications \cite{Sandry1,Sandry2,pp_graph2,blinddeconvgraph1,armagraph}.
\vspace{-0.in}
\subsection{The conventional analytic signal}
\indent Let $\hat x(\omega)$ denote the DFT of the real 1D signal $\mathbf{x}$ evaluated at frequency $\omega$. Then, the discrete analytic signal of $\mathbf{x}$, denoted by $\mathbf{x}_{a,c}$, has the following frequency-domain definition \cite{Oppenheim, Gold, Gabor}:
\begin{equation}
\label{AS}
\hat x_{a,c}(\omega)=\small\begin{cases}
2\hat x(\omega), \,\, \omega\in\left\{\textstyle\frac{2\pi}{N}, \cdots, \pi-\frac{2\pi}{N}\right\}\\
\hat x(\omega), \,\, \omega\in\left\{0,\pi\right\}\\
0, \,\,  \omega\in\left\{\pi+\frac{2\pi}{N}, \cdots, \frac{2\pi(N-1)}{N}\right\}.
\end{cases}
\end{equation}
Taking the inverse DFT on both sides of (\ref{AS}), we get that $\mathbf{x}_{a,c}=\mathbf{x}+\mathrm{j} \mathbf{x}_{h,c}$, where $\mathrm{j}=\sqrt{-1}$ and $\mathbf{x}_{h,c}$ is known as the discrete Hilbert transform of $\mathbf{x}$ \cite{Oppenheim}. The graph Hilbert transform has the following frequency-domain specification:
\begin{equation}
\label{AS2}
\hat x_{h,c}(\omega)=\small\begin{cases}
-\mathrm{j}\hat x(\omega), \,\, \omega\in\left\{\textstyle\frac{2\pi}{N}, \cdots, \pi-\frac{2\pi}{N}\right\}\\
\hat x(\omega), \,\, \omega\in\left\{0,\pi\right\}\\
+\mathrm{j}\hat x(\omega), \,\,  \omega\in\left\{\pi+\frac{2\pi}{N}, \cdots, \frac{2\pi(N-1)}{N}\right\}.
\end{cases}
\end{equation}

\section{Graph Analytic Signal}
\label{sec:gas}
We next define an analytic signal for signals over graphs. In our analysis, we make the following assumptions:
\begin{itemize}
\item[(1)]
$\mathbf{A}$ is real and asymmetric with atleast one conjugate-pair of eigenvalues. 
 \item[(2)] The Jordan (or eigen) decomposition of $\mathbf{A}$ is such that $\mathbf{J}$ has Jordan blocks arranged in the ascending order of phase angle of the eigenvalues from 0 to $2\pi$. If multiple eigenvalues with same phase angle occur, we order them in the descending order of their magnitude. 
\end{itemize}
We recall that the eigenvalues of a real-valued matrix, and the corresponding eigenvectors or generalized eigenvectors  are either real-valued or occur in complex-conjugate pairs\cite{Horn}.
Let $K_1$ and $K_2$ denote the number of real-valued positive and negative eigenvalues including repeated eigenvalues of $\mathbf{A}$, respectively, and $K=K_1+K_2$. Let us define the sets:
\begin{align}
\Gamma_1&=\left\{1, \cdots, K_1\right\} \hfill (\mbox{positive real eigenvalues}),\nonumber\\
\Gamma_2&=\left\{ K_1+1, \cdots, K_1+{\scriptstyle\frac{N-K}{2}}\right\} \,\,(\mbox{eigenvalues with phase angle in } (0,\pi)),\nonumber\\
\Gamma_3&=\left\{K_1+{\scriptstyle\frac{N-K}{2}}+1, \cdots, {\scriptstyle\frac{N+K}{2}}\right\}(\mbox{negative real  eigenvalues}),\,\,\nonumber\\
\Gamma_4&=\left\{{\scriptstyle\frac{N+K}{2}}+1, \cdots, N\right\}\,\,(\mbox{eigenvalues with phase angle in } (\pi,2\pi)),\nonumber
\end{align}
 and denote the vector spaces spanned by the corresponding eigenvectors by $\mathbf{V}_1$, $\mathbf{V}_2$, $\mathbf{V}_3$, and $\mathbf{V}_4$, respectively. For example, $\mathbf{V}_1$ is the space spanned by the eigenvectors related to $\Gamma_1$. On ordering as per Assumption 2, we have that for every $i$th eigenvector $\mathbf{A}$ such that $i\in\Gamma_2$, there exists an eigenvector indexed by $i\in\Gamma_4$ that share a complex-conjugate relationship, that is,
 	\begin{equation}
 	\mathbf{v}_i=
 	%
 	\mathbf{v}^*_{i'}, \,\,\,  i\in\Gamma_2,\,\,  i'\in\Gamma_4. \nonumber
 \end{equation}
 	In particular, for the case of all distinct eigenvalues, we have that 
 \begin{equation}
 \mathbf{v}_i=
 %
 \mathbf{v}^*_{(N-i+K_1+1)}, \,  i\in\Gamma_2. \nonumber
 \end{equation} 
Then, as a consequence of the complex-conjugate relationship, we have:
\begin{equation}
\hat x(i)=
\hat x^*(i'), \,\,\,  i\in\Gamma_2,\,i'\in\Gamma_4. \label{conj1}
\end{equation}
For real-valued $\mathbf{A}$, $N$ and $K$ are always of the same parity (odd or even). In the case of 1D signals, (\ref{conj1}) reduces to the conjugate-symmetry property of the DFT\cite{Oppenheim}. Equation (\ref{conj1}) indicates that a real graph signal can be represented using $\theta$ GFT coefficients, where $\theta=|\Gamma_1|+|\Gamma_2|+|\Gamma_3|=(N+K)/2$, and $|\Gamma|$ denotes the cardinality of the set $\Gamma$. For $K\ll N$, $\theta\approx N/2$. 
We note that (\ref{conj1}) holds only if $\mathbf{x}$ is real, which means that a graph signal which does not satisfy (\ref{conj1}) is necessarily complex-valued. Motivated by (\ref{conj1}) and conventional analytic signal construction, we next define the graph analytic signal  and graph Hilbert transform.
\begin{definition}
{
We define the graph analytic signal of $\mathbf{x}$ as  
$\mathbf{x}_a~=~\mathbf{V}\hat{ \mathbf{x}}_a$,
 where 
\begin{equation} 
 \hat{{x}}_a(i)~=\small~\begin{cases}2\hat x(i), &i\in\Gamma_2\nonumber\\
\hat x(i), &i\in\Gamma_1\cup\Gamma_3\nonumber\\
0,&i\in \Gamma_4
\end{cases}.
\end{equation}
As a consequence of the 'one-sidedness' of the GFT spectrum, we have that $\mathbf{x}_a$ is complex and hence, is expressible as $\mathbf{x}_a=\mathbf{x}+ \mathrm{j}\,\mathbf{x}_h$. We define $\mathbf{x}_h$ as the graph Hilbert transform of $\mathbf{x}$ such that 
\begin{equation}
\label{graphHT}
\mathrm{j}\hat{{x}}_h(i)=\small\begin{cases} +\hat x(i), &i\in\Gamma_2\\
0, &i\in\Gamma_1\cup\Gamma_3\\
-\hat x(i),&i\in \Gamma_4
\end{cases}.
\end{equation}
}
\end{definition}
On setting $\mathbf{A}=\mathbf{C}$, we observe that (\ref{graphHT}) reduces to the conventional Hilbert transform/analytic signal definitions given by (\ref{AS}) and (\ref{AS2}), that is, $\mathbf{x}_a=\mathbf{x}_{a,c}$ and $\mathbf{x}_h=\mathbf{x}_{h,c}$ since the $i$th graph frequency is equal to $e^{\mathrm{j}\omega_i}$ where $\omega_i\in\left\{0,\textstyle\frac{2\pi}{N}, \cdots, \pi\right\}$. This corresponds to $\Gamma_1=\{1\}$, $\Gamma_2=\{2,\cdots,N/2-1\}$, $\Gamma_3=\{N/2\}$, and $\Gamma_3=\{N/2+1,\cdots,N-1\}$ for even $N$. For odd $N$, this corresponds to $\Gamma_1=\{1\}$, $\Gamma_2=\{2,\cdots,(N+1)/2\}$, $\Gamma_3=\{\}$, and $\Gamma_4=\{(N+1)/2+1,\cdots,N-1\}$.

As an illustration of the graph analytic signal construction,  consider a graph with an adjacency matrix with eigenvalues distributed according to Figure \ref{GHTschematic}(a). Let us consider a signal $\mathbf{x}$ having unit GFT magnitude for all the graph frequencies. Then, graph Hilbert transform of $\mathbf{x}$ has the GFT spectrum shown in Figure \ref{GHTschematic}(c) since $\Gamma_1=\{1\}$, $\Gamma_2=\{2, 3 ,4\}$, $\Gamma_3=\{5\}$, and $\Gamma_4=\{6, 7 ,8\}$. The corresponding graph analytic signal has the GFT spectrum shown in Figure \ref{GHTschematic}(d).
\begin{figure}[t]
{\vspace{-0.in}
\centering
\subfigure[]{
\includegraphics[width=1.8in]{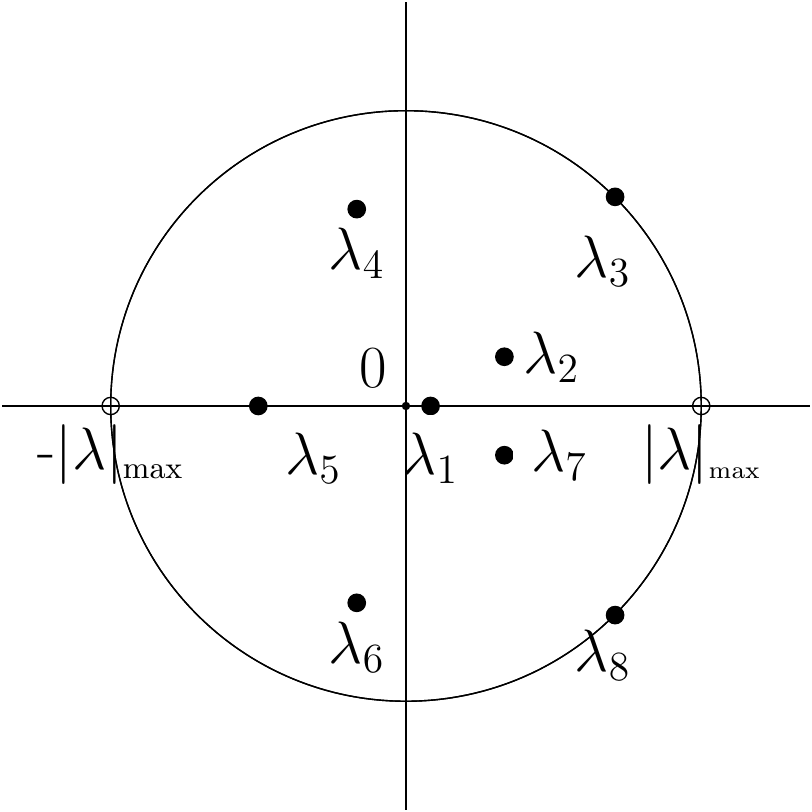}
}
\subfigure[]{
\includegraphics[width=2in]{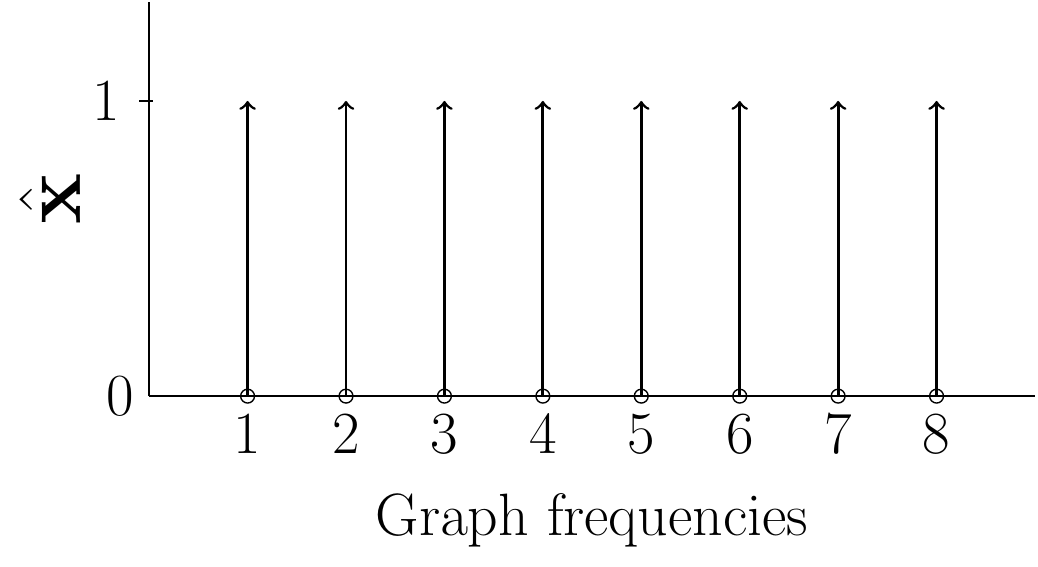}
}\hspace{-.0in}\\
\subfigure[]{
\includegraphics[width=1.8in]{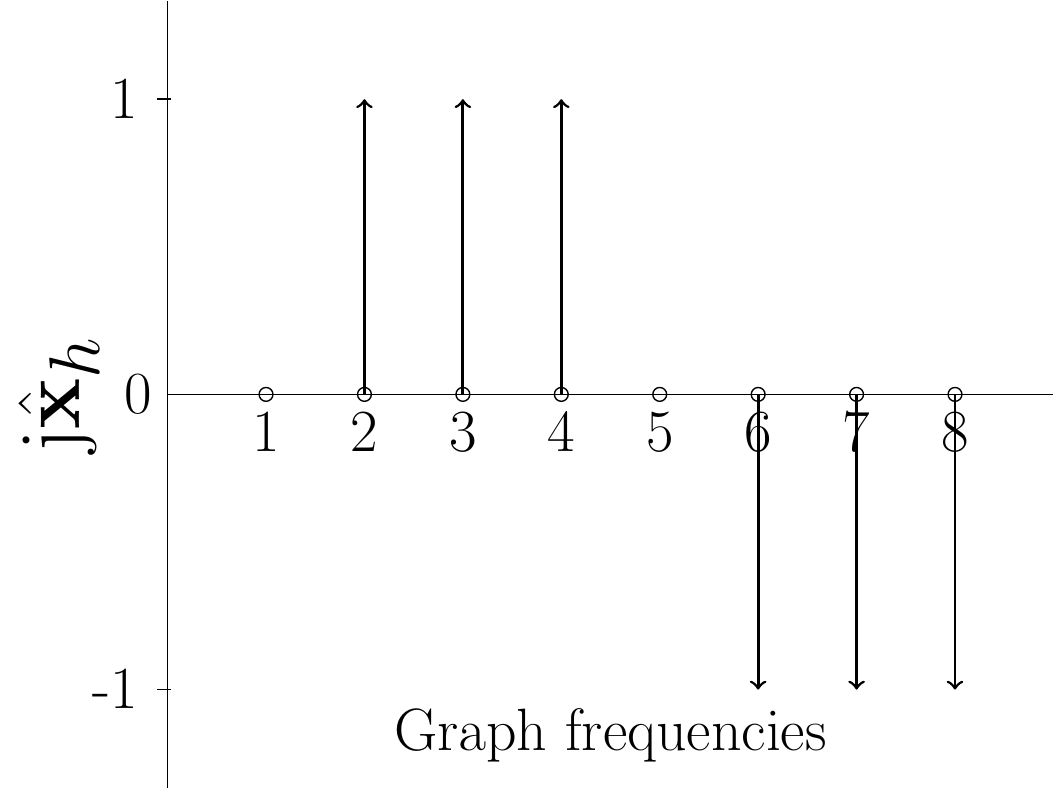}
}\hspace{-.1in}
\subfigure[]{
\includegraphics[width=2in]{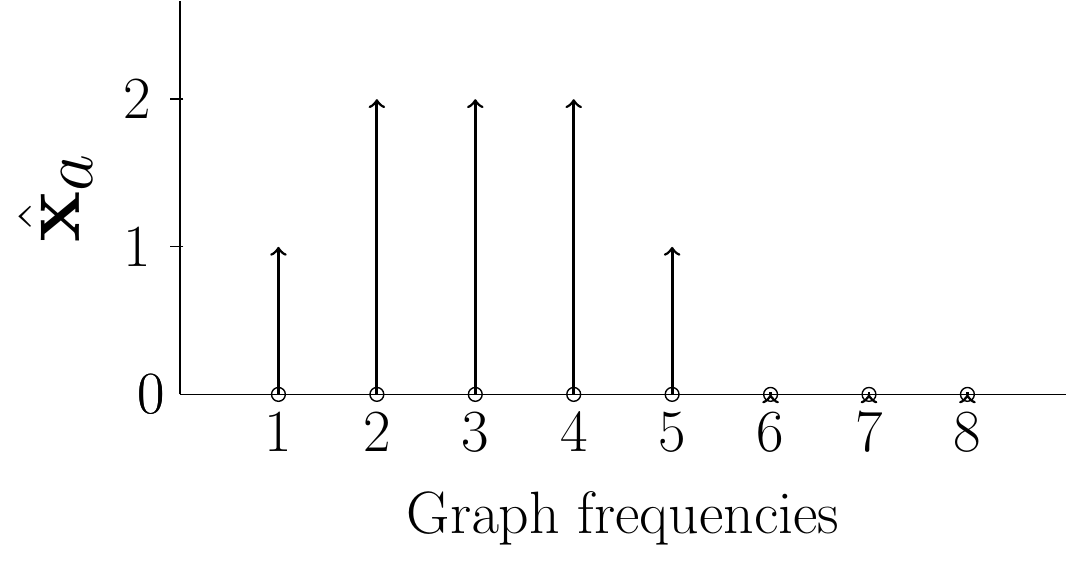}
}
\caption{ Illustration of graph analytic signal and graph Hilbert transform. (a) Eigenvalues of $\mathbf{A}$, (b) GFT of signal $\mathbf{x}$, (c) GFT of $\mathrm{j}\mathbf{x}_h$, and (d) GFT of $\mathbf{x}_a$. In this case,  $\Gamma_1=\{1\}$, $\Gamma_2=\{2, 3 ,4\}$, $\Gamma_3=\{5\}$, and $\Gamma_4=\{6, 7 ,8\}$.}
\label{GHTschematic}}
\end{figure}

\subsection{One-sided spectrum of the graph analytic signal}
The exact number of nonzero values in the graph analytic signal depends on the adjacency matrix $\mathbf{A}$. In the case when all the eigenvalues of $\mathbf{A}$ are complex $(K=0)$, the number of non-zero coefficients in $\hat{\mathbf{x}}_a$ is exactly one half of the total resulting in a one-sided spectrum, that is, $\theta=N/2$. We list the $\theta$ values for the 1D graph and random graphs Table \ref{onesidedtable}. For asymmetric matrices with entries drawn from independently and identical distributed (IID) mean zero unit variance Gaussian distribution $\mathcal{N}(0,1)$, the fraction of real eigenvalues asymptotically tends to zero \cite{Adelman1}. This was also shown to hold experimentally for matrices with independent and identically districuted entries from the uniform distribution over $[-1,1]$: $\mathcal{U}[-1,1]$, and Bernoulli $\{-1,1\}$ entries\cite{Adelman1,Tao}.  We note here that adjacency matrix with Bernoulli entries represents the Erd\H{o}s R\'{e}nyi model for small-world graphs\footnote{The Erd\H{o}s R\'{e}nyi model is a popular model for  small-world random graphs \cite{ErdosRenyi}.}. This implies that the corresponding graphs with adjacency matrices drawn from these distributions asymptotically have one-sided graph analytic signal spectrum. Since the Gaussian random matrix is a good approximation to general random matrices in terms of spectral properties, one can conclude that, on an average, asymmetric matrices have mostly  complex-valued eigenvalues .  This in turn indicates that most directed graphs have graph analytic signal with approximately 'one-sided' spectrum. 
%
\vspace{-0.in}
\begin{table}
	\centering
	\begin{tabular}{|c|c|c|c|}
		\hline
		Graph & Number of & $\theta$\\
		&real eigenvalues &\\
		\hline
		1D graph, odd $N$ & 1 &$\frac{N+1}{2}$\\
		1D graph, even $N$ & 2 &$\frac{N+2}{2}$\\
		Entries drawn  $\mathcal{N}(0,1)$ or & $\sqrt{\frac{2N}{\pi}}$ (asymptotic & $\frac{N}{2}+\sqrt{\frac{N}{2\pi}}$\\
		$\mathcal{U}[-1,1]$or Bernoulli $\{-1,1\}$&expected value) &\\
		\hline
	\end{tabular}
	\hspace{.1in}
	\caption{$\theta$ value for some graphs of interest.}
	\label{onesidedtable}
	\vspace{-.0in}
\end{table}

\subsection{Discussions on graph analytic signal and graph Hilbert transform}
\label{GHT_sec}
We next show that the graph Hilbert transform $\mathbf{x}_h$ of a real graph signal $\mathbf{x}$ is real. Since 
$\mathbf{x}_h=\mathbf{V} \hat{\mathbf{x}}_h$, we have that
\begin{align}
\label{x_h_real}
\mathrm{j}{ \mathbf{x}}_h&=\textstyle \sum_{i\in\Gamma_2}\mathrm{j}\hat x_h(i) \mathbf{v}_i+\textstyle \sum_{i\in\Gamma_4}\mathrm{j}\hat x_h(i) \mathbf{v}_i =\textstyle \sum_{i\in\Gamma_2}\hat x(i) \mathbf{v}_i-\textstyle \sum_{i\in\Gamma_4}\hat x(i) \mathbf{v}_i \\
&=\textstyle \sum_{i\in\Gamma_2}\left(\hat x(i) \mathbf{v}_i-\hat x^*(i) \mathbf{v}^*_i\right)=2\,\mathrm{j}\Im\left(\textstyle \sum_{i\in\Gamma_2} \hat x(i) \mathbf{v}_i\right),\nonumber
\end{align}
where $\mathbf{v}_i$ denotes the $i$th column of $\mathbf{V}$, 
 and $\Im(a)$ denotes the imaginary part $a$. The third equality in (\ref{x_h_real}) follows because the eigenvectors indexed by $\Gamma_2$ and $\Gamma_4$ form complex conjugates. Thus, $\mathrm{j}\mathbf{x}_h$ is purely imaginary which in turn means that $\mathbf{x}=\Re(\mathbf{x}_a)$, where $\Re(a)$ denotes the real part of $a$. 
We express \eqref{graphHT} as 
\begin{eqnarray}
\hat {\mathbf{x}}_h=\mathbf{J}_h\hat{\mathbf{x}},\,\,\mbox{or} \,\,\,
\mathbf{x}_h\triangleq\mathcal{H}\{\mathbf{x}\}=\mathbf{V}\mathbf{J}_h\mathbf{V}^{-1}\mathbf{x},
\label{graphHTvec}
\end{eqnarray}
where $\mathbf{J}_h$ is the diagonal matrix with $i${th} diagonal element:
\begin{equation}
\label{graphHT_freq}
J_h(i)=\small\begin{cases}-\mathrm{j}, &i\in\Gamma_2\\
0, &i\in\Gamma_1\cup\Gamma_3.\\
+\mathrm{j},&i\in \Gamma_4
\end{cases}
\end{equation} 
 \begin{proposition}
	\label{proposition1}
{The graph Hilbert transform is a linear shift-invariant graph filtering operation for diagonalizable graphs.}
\end{proposition}
\begin{proof}
From (\ref{graphHTvec}), we have that $\mathbf{x}_h=\mathbf{V}\mathbf{J}_h\mathbf{V}^{-1}\mathbf{x}=\mathbf{Hx}$, where $\mathbf{H}=\mathbf{V}\mathbf{J}_h\mathbf{V}^{-1}$. By definition, graph filter $\mathbf{H}$ is linear and shift-invariant if for any graph filter of the form $\mathbf{M}=\sum_{i=0}^M m_i\mathbf{A}^i=m(\mathbf{A})$, $M\leq N$ we have
$\mathbf{H}\mathbf{M}\mathbf{x}=\mathbf{M}\mathbf{Hx},$
which in turn means that $\mathbf{H}$ should be a polynomial of $\mathbf{M}$, or equivalently, of $\mathbf{A}$. Since $\mathbf{A=VJV^{-1}}$, we have that $\mathbf{M}=\mathbf{V}m(\mathbf{J})\mathbf{V}^{-1}$. Let $\mathbf{y}$ denote the output of filter $\mathbf{M}$ for the input $\mathbf{x}$: $\mathbf{y = Mx}$. Then,  we have that $\hat {\mathbf{y}} \triangleq \mathbf{V}^{-1} \mathbf{y} = \mathbf{V}^{-1} \mathbf{Mx} =m(\mathbf{J})\hat{ \mathbf{x}}$. 
Since $\hat {\mathbf{x}}_h=\mathbf{J}_h\hat{\mathbf{x}}$, we get that
\begin{equation}
\hat{\mathbf{ y}}_h=\mathbf{J}_h\hat{\mathbf{y}}=\mathbf{J}_h\, m(\mathbf{J})\,\hat{\mathbf{x}}=m(\mathbf{J})\,\mathbf{J}_h\, \hat{\mathbf{x}}=m(\mathbf{J})\,\hat{ \mathbf{x}}_h, \label{jj}
\end{equation} 
where we use the commutativity of the diagonal matrices $m(\mathbf{J})$ and $\mathbf{J}_h$ ($\mathbf{J}$ is a diagonal because $\mathbf{A}$ is diagonalizable.) Taking inverse GFT on both sides of (\ref{jj}), we get that 
$\mathbf{y}_h \triangleq \mathbf{H}m(\mathbf{A})\mathbf{x}=m(\mathbf{A})\mathbf{x}_h=m(\mathbf{A})\mathbf{H}\mathbf{x}.\nonumber
$
\end{proof} 
 The graph Hilbert transform being a shift-invariant filter means that there exists a polynomial $h(x)=\sum_{i=0}^L h_i{x}^i$ such that $\mathbf{H}=h(\mathbf{A})$. The coefficients are evaluated by noting that $\mathbf{H}$ modifies the spectrum of the graph signal with constant spectrum of unit amplitude as specified in \eqref{graphHT}. In other words, we solve for $h_i$s such that the $i$th component $c(i)$ of the vector $\mathbf{c}=h(\mathbf{J})\mathbf{1}_N$ is given by
$c(i)=\small\begin{cases}
0,\quad i\in\Gamma_1\cup\Gamma_3\nonumber\\
 -\mathrm{j}, \quad i\in\Gamma_2\nonumber\\
 +\mathrm{j} \quad i\in\Gamma_4.\nonumber
 \end{cases} $.
 Since $c_i=h(\lambda_i)$ in the case of a diagonal $\mathbf{J}$, $h_i$s  are obtained by solving:
\begin{eqnarray}
\label{ght_filter}
 \small   h_{0} +    h_{1}\lambda_i  + \cdots   + h_L\lambda^L_{i} &=& 0,\quad i\in\Gamma_1\cup\Gamma_3\nonumber\\
     h_{0} +    h_{1}\lambda_i  + \cdots   + h_L\lambda^L_{i} &=& -\mathrm{j},\quad i\in\Gamma_2\\
     h_{0} +  h_{1}\lambda_i  + \cdots   + h_L\lambda^L_{i} &=&+\mathrm{j},\quad i\in\Gamma_4.\nonumber
  \end{eqnarray}
The solution of (\ref{ght_filter}) obtained by setting $\mathbf{A=C}$ and $L=N$ is the impulse response of the discrete Hilbert transform. In order to avoid ill-conditioning of (\ref{ght_filter}), $L$ is usually restricted to be much less than $N$.
In Figure \ref{GHT_IR}, we show the graph Hilbert transform computed using (\ref{ght_filter}) for various values of $L$ for the 1D signal graph. We observe from Figure \ref{GHT_IR}(d) that as $L$ is decreased, the spectrum of the graph Hilbert transform differs from the ideal case. This is because the corresponding columns of each graph shift ($\mathbf{A}^i$) are linearly independent and restricting the number of taps restricts the dimension of the signal space. 
\begin{figure}[t]
	\vspace{-.2in}
	\centering
	$
	\begin{array}{cc}
	\subfigure[]{\hspace{-.0in}
		\includegraphics[width=2in]{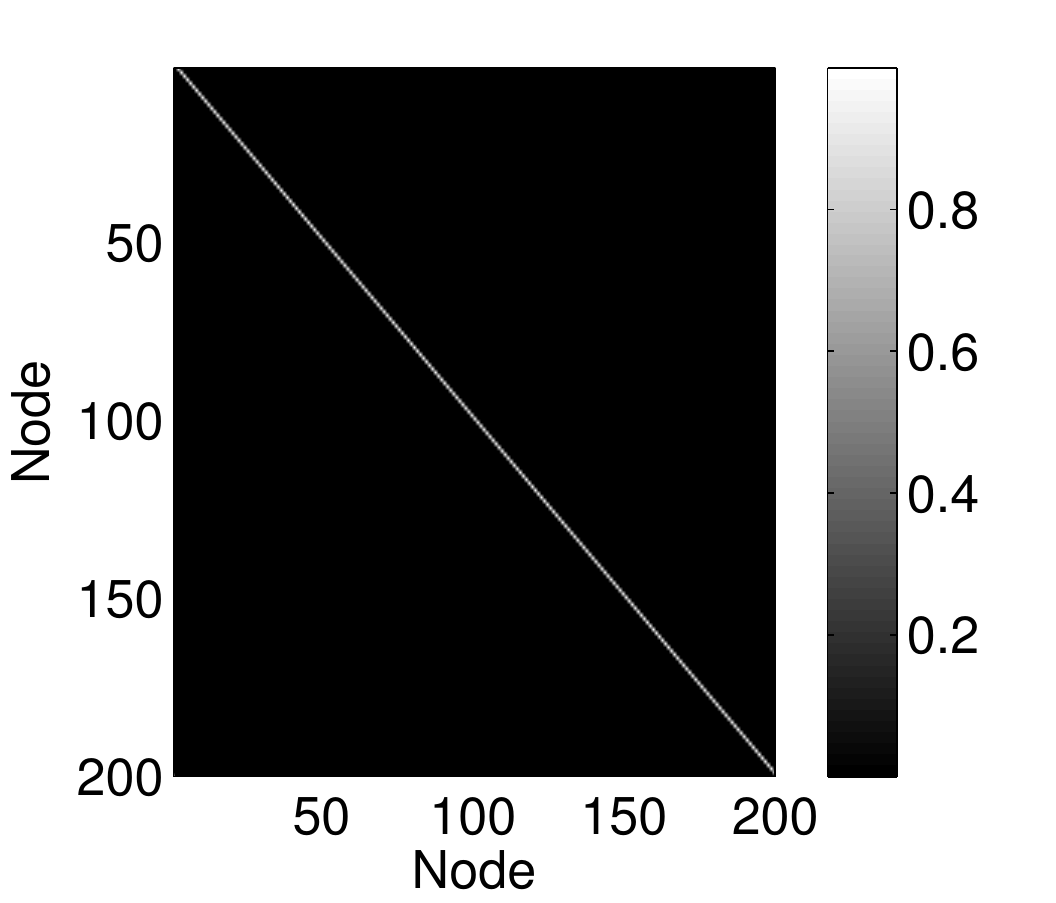}
	}\hspace{-.0in}
	\subfigure[]{
		\includegraphics[width=2in]{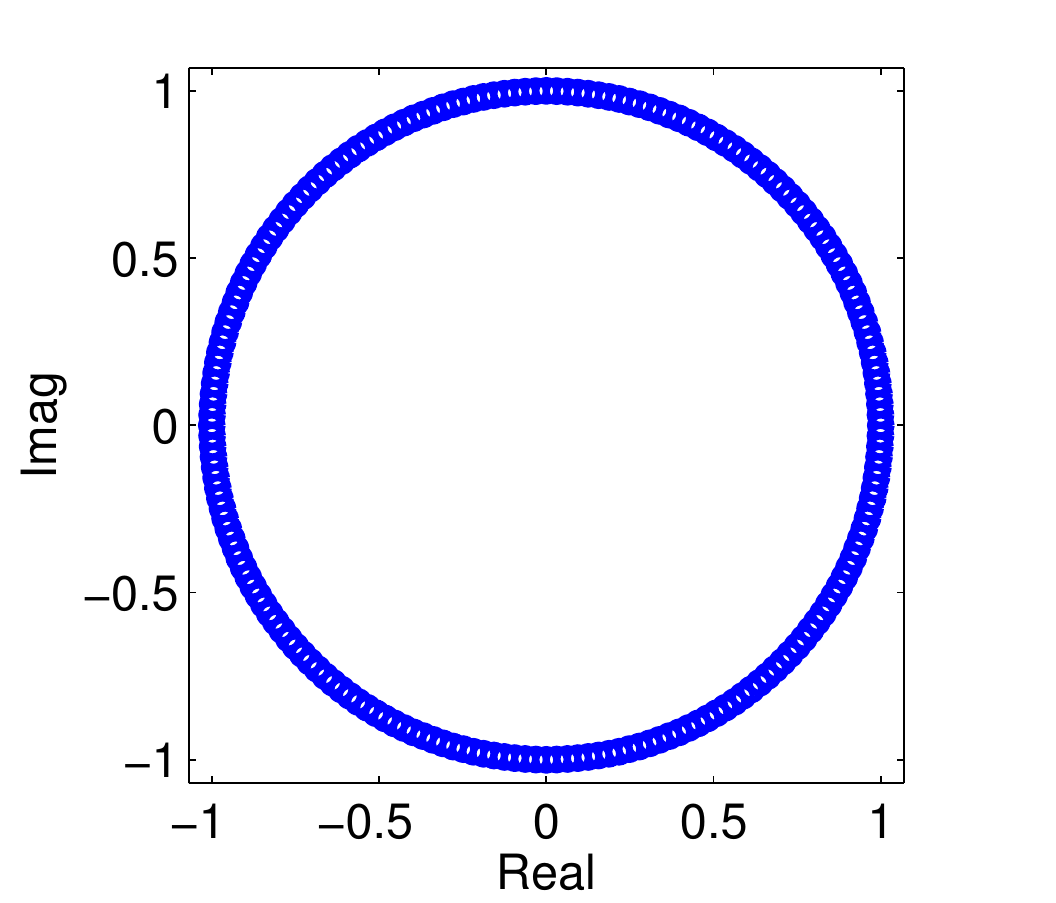}
	}\\
	\hspace{-.0in}
	\subfigure[]{
		\includegraphics[width=2in]{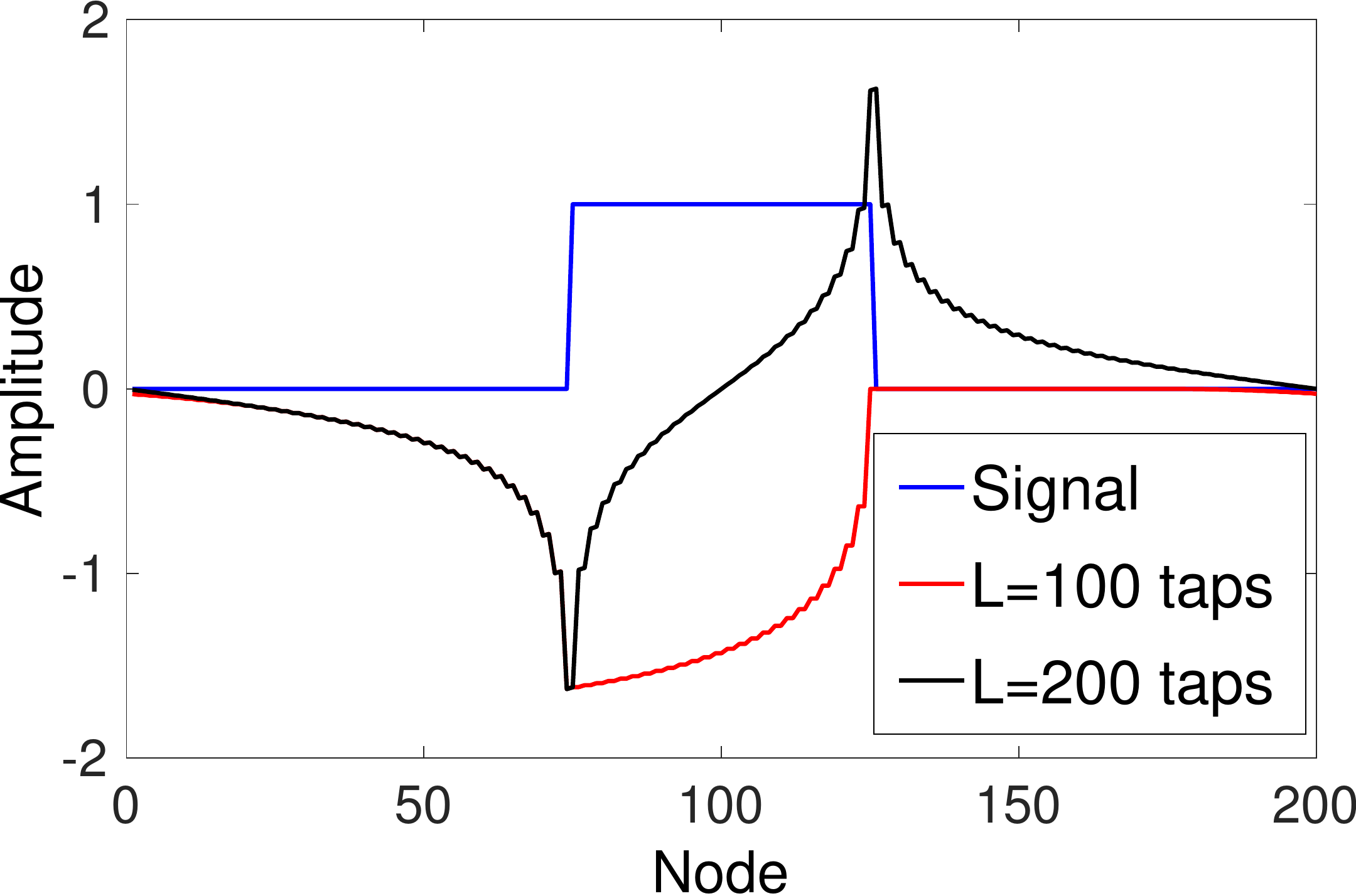}
	}
	\subfigure[]{\hspace{-.0in}
		\includegraphics[width=2in]{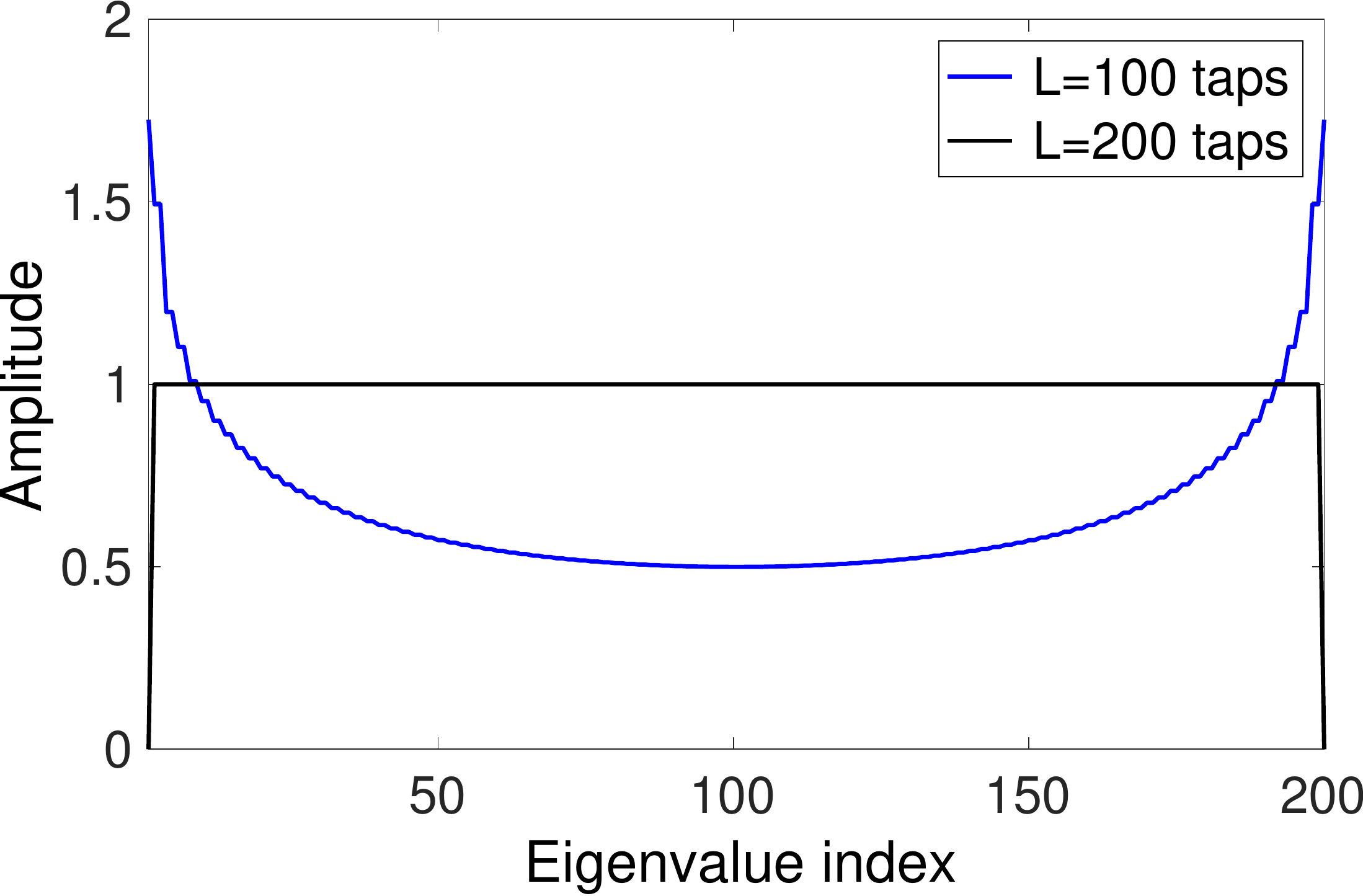}
		\hspace{.1in}
	}
	\end{array}
	$
	\caption{1D signal graph (a) $\mathbf{A}$, (b) Eigenvalues of $\mathbf{A}$, (c) Signal and its graph Hilbert transform computed using $L=N$ and $L=N/2$, (d) GFT spectrum of graph Hilbert transform.
	}
	\label{GHT_IR}
\end{figure}
\vspace{-.in}
\subsection{Some properties of graph analytic signal/graph Hilbert transform}
 Let $\mathcal{A}$, $\mathcal{I}$, and $\mathcal{H}$ denote the graph analytic signal, identity, and graph Hilbert transform operators, respectively, such that $\mathbf{x}_a=\mathcal{A}\{\mathbf{x}\},\,\, \mathbf{x}_h=\mathcal{H}\{\mathbf{x}\}$ and $\mathcal{A}=\mathcal{I}+\mathrm{j}\mathcal{H}$. For $\mathbf{x}= \eta\mathbf{f}+ \eta^*\mathbf{f}^*$ such that $\mathbf{f}\in\mathbf{V}_2\cup\mathbf{V}_4$ and $\eta\in\mathbb{C}$, we have the following properties:
\begin{enumerate}[leftmargin=0.5cm]
\item {\it Graph-shift invariance:} $\mathcal{H}\{\alpha \mathbf{Ax}\}=\alpha\mathbf{A}\mathcal{H}\{\mathbf{x}\}$, $\alpha\in\mathbb{C}$.
\item{\it Superposition:} For $\mathbf{x}_1,\mathbf{x}_2\in \mathbf{V}_2\cup\mathbf{V}_4$ and $\alpha,\beta\in\mathbb{C}$, we have that $\mathcal{H}\{\alpha\mathbf{x}_1+\beta\mathbf{x}_2\}=\alpha\mathcal{H}\{\mathbf{x}_1\}+\beta\mathcal{H}\{\mathbf{x}_2\}$.\\
	{\it Proof}:
From (\ref{graphHTvec}), we have that for $\mathbf{x}=\alpha\mathbf{x}_1+\beta\mathbf{x}_2$
\begin{align}
\mathcal{H}\{\mathbf{x}\}&=\mathbf{V}\mathbf{J}_h\mathbf{V}^{-1}\{\mathbf{x}\}
=\mathbf{V}\mathbf{J}_h\{\alpha\mathbf{V}^{-1}\mathbf{x}_1+\beta\mathbf{V}^{-1}\mathbf{x}_2\}\nonumber\\
&=\mathbf{V}\mathbf{J}_h\{\alpha\hat{\mathbf{x}}_1+\beta\hat{\mathbf{x}}_2\}=\mathbf{V}\{\alpha\hat{\mathbf{x}}_{1,h}+\beta\hat{\mathbf{x}}_{2,h}\}=\alpha\mathcal{H}\{\mathbf{x}_{1}\}+\beta\mathcal{H}\{\mathbf{x}_{2}\},\nonumber
\end{align}
where $\hat{\mathbf{x}}_{1,h}$ and $\hat{\mathbf{x}}_{2,h}$ denote the GFT of $\mathcal{H}\{\mathbf{x}_{1}\}$ and $\mathcal{H}\{\mathbf{x}_{2}\}$, respectively.

\item {\it Phase-shifting action\footnote{For simplicity, we use the same operator notation to denote the corresponding operation for both the signal seen as a vector and as a function of the node. For example, $\mathcal{H}\{ \cos(\omega n)\}$ denotes the operator action directly on the function $\cos(\omega n)$ evaluated at the $n$th node, whereas $\mathcal{H}\{\mathbf{x}\}$ denotes the vector that comes out of applying the graph Hilbert transform operation on the signal $\mathbf{x}$. In the case when $x(n)=\cos(\omega n)$, we have $\mathcal{H}\{\mathbf{x}\}(n)=\mathcal{H}\{ \cos(\omega n)\}$.}:} For $i\in\Gamma_2\cup\Gamma_4$: 
\begin{eqnarray}
\mathcal{H}\{\Re\left(\mathbf{v}_i\right)\}=\Im\left(\mathbf{v}_i\right),\,\,\mbox{and}\,\,
\mathcal{H}\{\Im\left(\mathbf{v}_i\right)\}=-\Re\left(\mathbf{v}_i\right).
 \label{realimag}
\end{eqnarray} 
\emph{Proof:}  
Consider $i\in\Gamma_2$. Using the property 2, we have 
\begin{align}
\mathcal{H}\{2\Re\left(\mathbf{v}_i\right)\}&=\mathcal{H}\{\mathbf{v}_i + \mathbf{v}^*_i\}=\left(\mathrm{j}\mathbf{v}_i-\mathrm{j}\mathbf{v}^*_i \right)=2\Im\left(\mathbf{v}_i\right)\nonumber\\
\mathcal{H}\{2\mathrm{j}\Im\left(\mathbf{v}_i\right)\}&=\mathcal{H}\{\mathbf{v}_i - \mathbf{v}^*_i\}=\left(\mathrm{j}\mathbf{v}_i+\mathrm{j}\mathbf{v}^*_i \right)=2\mathrm{{j}}\Re\left(\mathbf{v}_i\right).
\nonumber
\end{align}
The proof for $i\in\Gamma_4$ follows similarly. 
Equation (\ref{realimag}) generalizes the quadrature phase-shifting action of the discrete Hilbert transform $\mathcal{H}_c$ on sinusoids: 
\begin{align}
\mathcal{H}_c\{\cos (\omega_i n)\}&=\mathcal{H}_c\{\Re\{\mathbf{v}_i\}\}(n)=\sin (\omega_i n)=\Im\{\mathbf{v}_i\}(n),\nonumber\\
\mathcal{H}_c\{\sin (\omega_i n)\}&=\mathcal{H}_c\{\Im\{\mathbf{v}_i\}\}(n)= - \cos (\omega_i n)=-\Re\{\mathbf{v}_i\}(n), \nonumber
\end{align}
noting when $\mathbf{A=C}$, $\mathbf{v}_i=e^ {\mathrm{j}(\omega_i n)}$, where $\omega_i=\frac{ 2\pi(i-1)}{N}$. 
\item {\it Inverse:}
$ \mathcal{H}^2=-\mathcal{I}$ or, $\mathcal{H}^{-1}=-\mathcal{H}$.\\
\emph{Proof}: 
From (\ref{graphHTvec}), we have $\hat {\mathbf{x}}_h=\mathbf{J}_h\hat{\mathbf{x}}$. Hence, the GFT of $\mathcal{H}^2\{\mathbf{x}\}$ is given by $\mathbf{J}^2_h \mathbf{\hat{x}}$.  We have from (\ref{graphHT_freq}) that
$J^2_h(i)=J_h(i)J_h(i)=\small\begin{cases}\mathrm{j}^2=-1, &i\in\Gamma_2\cup\Gamma_4\\
0, &i\in\Gamma_1\cup\Gamma_3\nonumber
\end{cases},
\nonumber
$
which shows that $\mathbf{J}^2_h\mathbf{\hat{x}}=-\mathbf{\hat{x}}$ for $\mathbf{x}\in\mathbf{V}_2\cup\mathbf{V}_4$. In other words, $\mathcal{H}^2\{\mathbf{x}\}=-\mathbf{x}$~which completes the proof.

\item {\it Repeated operation:} $\mathcal{A}^2\{\mathbf{x}\}=(\mathcal{I+\mathrm{j}H})^2\{\mathbf{x}\}=2\mathcal{A}\{\mathbf{x}\}$ and $\mathcal{H}^4\{\mathbf{x}\}=\mathbf{x}$. (Follows from Property 4).

\item {\it Isometry: 
}Since $\forall i$, $|\hat{x}_h(i)|=|\hat{x}(i)|$, we have that $\|\hat{\mathbf{x}}_h\|_p=\|\hat{\mathbf{x}}\|_p$, where $\|\hat{\mathbf{x}}\|_p=\left(\sum_i |\hat{x}(i)|^p\right)^{\frac{1}{p}}$ is the $\ell_p$ norm of $\hat{\mathbf{x}}$, $p\geq 1$, assuming $\|\hat{\mathbf{x}}\|_p< \infty$, that is,  $\hat{\mathbf{x}}\in\ell^N_p(\mathbb{C})$. 
In particular, 
  $\|\hat{\mathbf{x}}_h\|^2_2=\|\hat{\mathbf{x}}\|^2_2=\frac{1}{2}\|\hat{\mathbf{x}}_a\|_2^2$. If  $\mathbf{V}$ is unitary, $ \|\mathbf{x}_h\|^2_2= \|\mathbf{x}\|^2_2=\frac{1}{2}\|\mathbf{x}_a\|^2_2$.
\item{\it Preservation of orthogonality}: {If $\mathbf{V}$ is unitary and $\textbf{x}_1$, $\textbf{x}_2$ are orthogonal, $\langle\mathbf{x}_1,\mathbf{x}_2\rangle=0$, then $\langle\mathcal{H}\mathbf{x}_1,\mathcal{H}\mathbf{x}_2\rangle=0$.}\\
\emph{Proof:}
Since $\mathbf{V}^{-1}$ is unitary, $\langle\hat{\mathbf{x}}_1,\hat{\mathbf{x}}_2\rangle=\langle\mathbf{x}_1,\mathbf{x}_2\rangle$. Then we have that 
\begin{align}
\langle\mathcal{H}\mathbf{x}_1,\mathcal{H}\mathbf{x}_2\rangle
&=\langle\mathbf{V}^{-1}\mathcal{H}\mathbf{x}_1,\mathbf{V}^{-1}\mathcal{H}\mathbf{x}_2\rangle
=\langle\mathbf{J}_h\hat{\mathbf{x}}_1,\mathbf{J}_h\hat{\mathbf{x}}_2\rangle= \langle\hat{\mathbf{x}}_1,\hat{\mathbf{x}}_2\rangle=\langle\mathbf{x}_1,\mathbf{x}_2\rangle=0.\nonumber
\end{align}
\end{enumerate} 
Since $\mathcal{A}=\mathcal{I}+\mathrm{j}\mathcal{H}$, properties 1 and 2 are also satisfied by $\mathcal{A}$. Properties 1 to 7 do not hold if $\mathbf{x}$ has contribution from the subspaces $\mathbf{V}_1$ or $\mathbf{V}_3$ as $\mathcal{H\{\mathbf{x}\}}=\mathbf{0}$ for $\mathbf{x}\in \mathbf{V}_1\cup\mathbf{V}_3$. We note here that the 'one-sidedness' and other properties of the graph analytic signal/graph Hilbert transform of a real signal $\mathbf{x} $ are decided entirely by the adjacency matrix $\mathbf{A}$ of the underlying graph. A real signal $\mathbf{x}$ may have a graph analytic signal $\mathbf{x}_a$ with a larger number of nonzero GFT coefficients in one graph than in another graph.
\subsection{On the graph Hilbert transform and the Bedrosian property}
\label{GHTbedrosian_sec}
We show a limitation of the graph Hilbert transform in this section.
One of the important properties possessed by the conventional Hilbert transform (and its fractional versions \cite{ArunFrHT}) is that it obeys the Bedrosian property, that is, if $f$ and $g$ are two signals with disjoint Fourier spectra (DFT or discrete-time FT, such that $f$ is low-pass and $g$ is high-pass, then we have that the Hilbert transform $\mathcal{H}_c$ satisfies
\begin{equation}
\mathcal{H}_c\{f(n)g(n)\}= f(n)\mathcal{H}_c\{g(n)\}.\quad \forall n
\nonumber
\end{equation}
As we show through experiments next, the graph Hilbert transform of a general graph does not possess the Bedrosian property. In our opinion, there are two important factors for this limitation. First, the point-wise product in the node domain does not correspond to a graph-frequency domain convolution (no definition for frequency domain convolution for graphs exists currently).  Second, unlike the 1D case where the DFT basis is functionally related to the frequency (as the entries of the $i$th DFT column are given by complex exponentiation of the $i$th frequency to different powers), the GFT of a general graph usually is not related to its eigenvalues through analytical expressions. We demonstrate this by considering the jittered 1D signal modeled as a graph signal with
$\mathbf{A}=\small\scriptstyle\left(\begin{array}{cccccc}
0& w_1& 0&\cdots& 0\\
0&0&w_2&\cdots&0\\\nonumber
\vdots&\vdots&\vdots &\vdots &\\
w_N&0&0&\cdots&0
\end{array}\right)$,
where $w_i$ denotes the spacing between the $i$th and $(i+1)$th samples. For uniformly sampled 1D signal, $w_i=1$ for all $i$. 
From the Bedrosian property we have that 
\begin{equation}
\mathcal{H}_c\{\cos(\omega_in)\cos(\omega_j n)\}= \cos(\omega_in)\sin(\omega_jn),\,\,\,\mbox{for}\,\, \omega_j>\omega_i
\nonumber
\end{equation}
which when expressed in terms of the corresponding DFT (GFT) vectors  $\mathbf{v}_i(n)=e^{\mathrm{j}\omega_i n}$ becomes
\begin{equation}
\mathcal{H}\{\Re\left(\mathbf{v}_i\right)\cdot \Re\left(\mathbf{v}_j\right)\}=\Re\left(\mathbf{v}_i\right)\cdot\Im\left(\mathbf{v}_j\right),
\label{bedrosian_prop}
\end{equation}
where $\mathbf{f}\cdot\mathbf{g}$ denotes the vector obtained by component-wise products of $\mathbf{f}$ and $\mathbf{g}$.
We test the validity of the Bedrosian property by computing the graph Hilbert transform of the signal $ \Re\left(\mathbf{v}_i\right)\cdot \Re\left(\mathbf{v}_j\right)$ where $i$ and $j$ correspond to the low and high frequency GFT basis vectors, respectively (Here we use the  frequency-ordering proposed in \cite{Sandry1} based on $MS_g$ $\|\mathbf{x}-\mathbf{Ax}\|_2$). We compare the graph Hilbert transform of $\Re\left(\mathbf{v}_i\right)\cdot \Re\left(\mathbf{v}_j\right)$ with $\Re\left(\mathbf{v}_i\right)\cdot\Im\left(\mathbf{v}_j\right)$. We consider the low-jitter case $w_i=1+d_i$, $d_i$s drawn independently from the Gaussian distribution $\mathcal{N}(0,0.01)$. Repeating the experiment multiple times, we observe that the left and right hand sides of (\ref{bedrosian_prop}) almost never coincide. In Figure \ref{bedrosian}, we consider a particular realization for $\mathbf{x}=\Re\left(\mathbf{v}_3\right)\cdot \Re\left(\mathbf{v}_{57}\right)$, where $N=100$. We observe that the graph Hilbert transform does not coincide with either $\Re\left(\mathbf{v}_3\right)\cdot \Im\left(\mathbf{v}_{57}\right)$ or $\Im\left(\mathbf{v}_3\right)\cdot \Re\left(\mathbf{v}_{57}\right)$. This shows that the graph Hilbert transform does not possess a Bedrosian property even for graphs approximately similar to the 1D graph. 
\begin{figure}[t]
\vspace{-.1in}
\centering
\subfigure[\hspace{-.2in}]{
\includegraphics[width=2in]{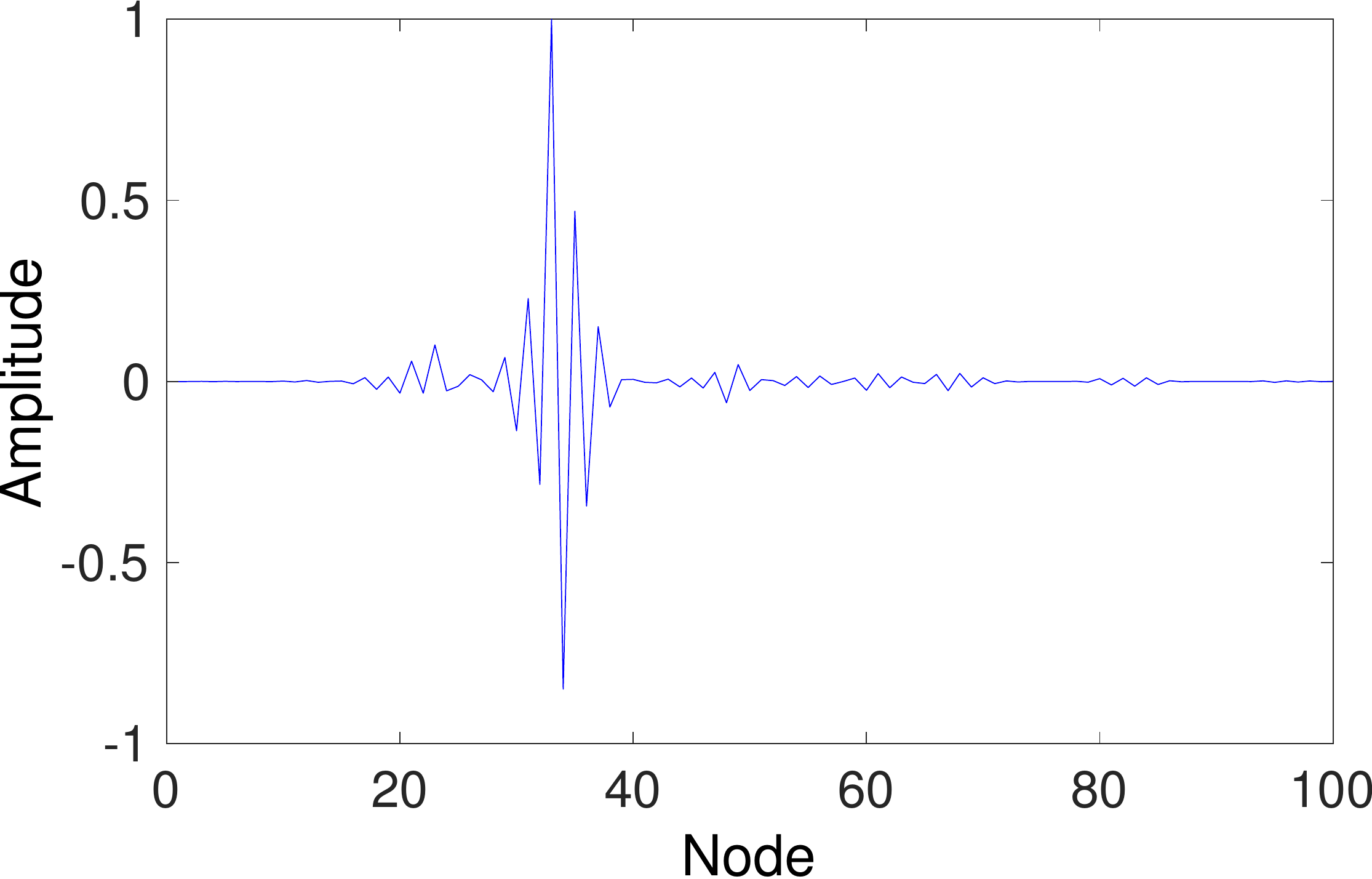}}\\
$
\begin{array}{rr}
\hspace{-.1in}
\subfigure[\hspace{-.2in}]{
\includegraphics[width=2in]{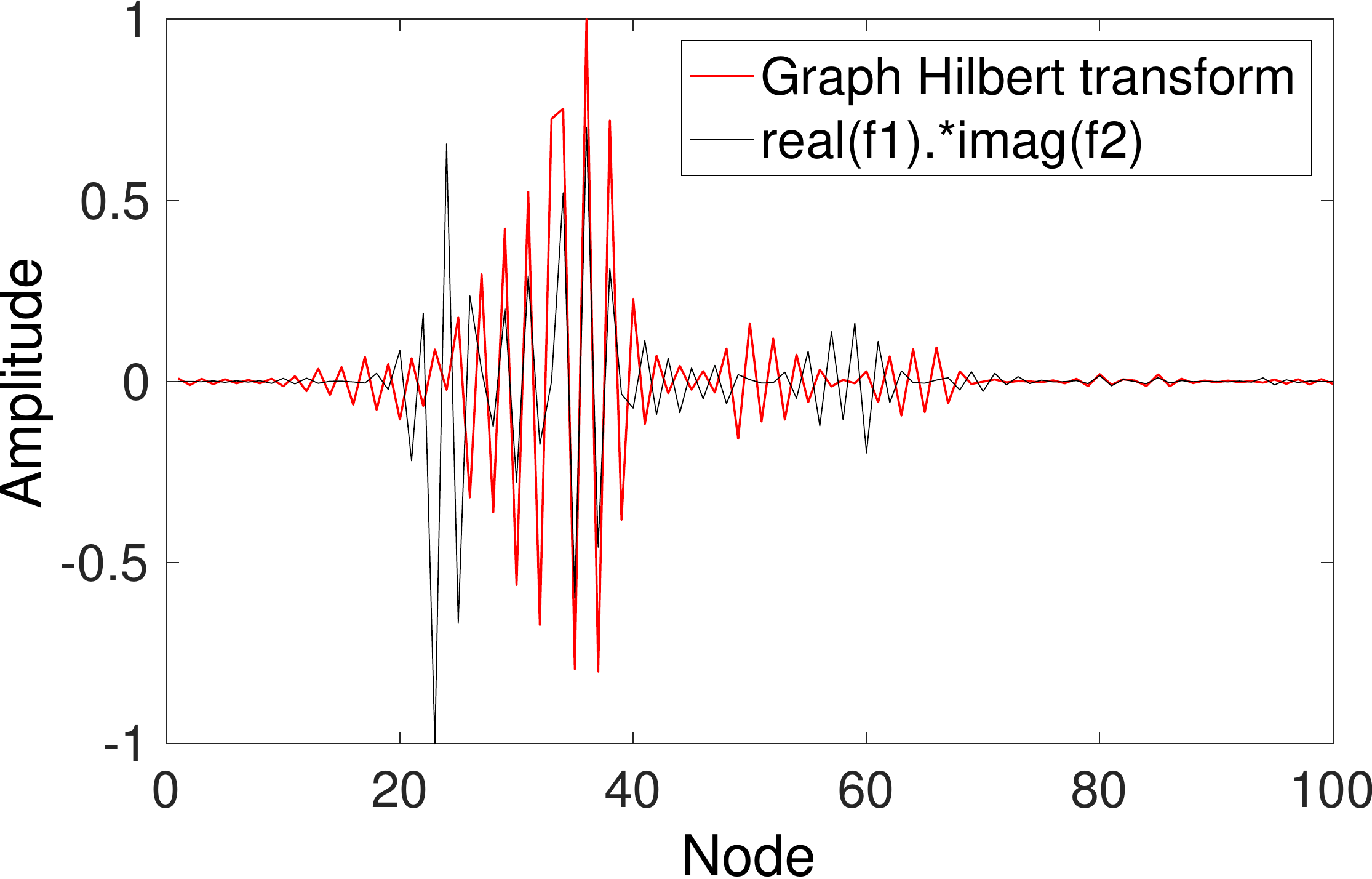}
}\hspace{-.in}
\subfigure[\hspace{-.1in}]{
\includegraphics[width=2in]{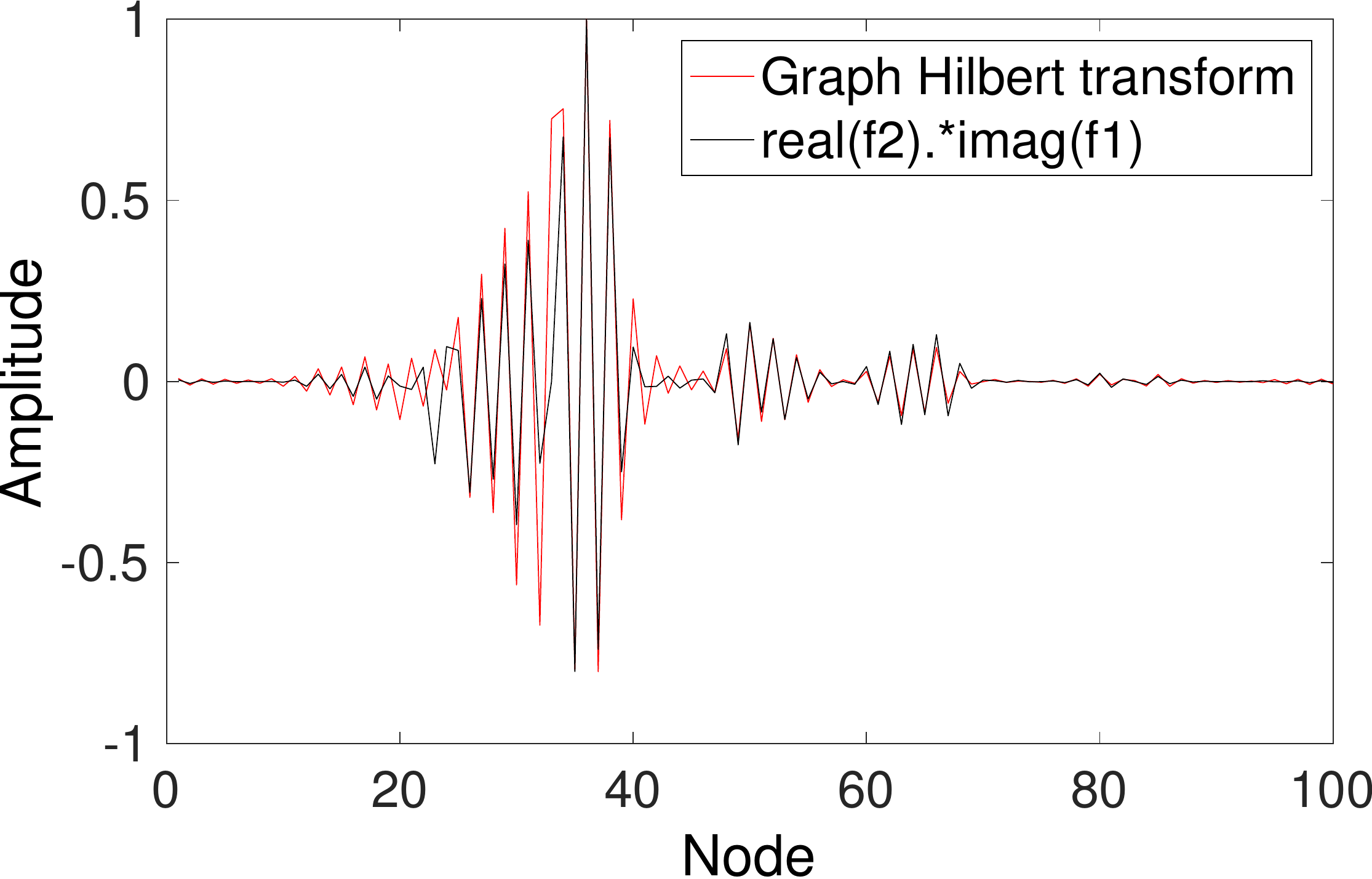}
}
\end{array}
$
\caption{Graph Hilbert transform of a signal for the jittered 1D signal graph $\mathbf{x}=\Re\left(\mathbf{f}_1\right)\cdot \Re\left(\mathbf{f}_{2}\right)$, where $\mathbf{f}_1=\mathbf{v}_3$ ({\color{black}$\mbox{MS}_g$} = 0.092) and $\mathbf{f}_2=\mathbf{v}_{57}$ ($\mbox{MS}_g$ = 0.521). (a) Signal, (b) $\mathcal{H}\{\mathbf{x}\}$ and $\Re\left(\mathbf{f}_1\right)\cdot \Re\left(\mathbf{f}_{2}\right)$, and (c)~$\mathcal{H}\{\mathbf{x}\}$ and $\Re\left(\mathbf{f}_2\right)\cdot \Re\left(\mathbf{f}_{1}\right)$. }
\label{bedrosian}
\end{figure}
 \vspace{-0.0in}
\section{The Graph analytic signal and Modulation Analysis}
\label{GAMFM}
The concept of analytic signal is used extensively in the demodulation of amplitude-modulated frequency-modulated signals\cite{ Boashash,Vakman_book,Maragos,Handel_Maragos,Arun_binaural}. Modulation analysis decomposes a signal into two components: one varying smoothly, capturing the average information in the signal (referred to as the AM), and the second, capturing the finer variations (referred to as the phase or frequency modulation (PM or FM)). Most demodulation techniques involve the construction of the analytic signal, implicitly or explicitly. Motivated by 1D modulation definitions \cite{Boashash,Vakman_book}, we next propose AM and PM for graph signals:
\begin{definition}
The AM $\mathcal{A}_{x,\mathcal{V}}$ and PM $\phi_{x,\mathcal{V}}$ of a graph signal $\mathbf{x}$ are defined as the magnitude and phase angle of the graph analytic signal, respectively:
\begin{eqnarray}
\mathcal{A}_{x,\mathcal{V}}(i)&=&\left| {x}_a(i)\right|,\quad  \forall\,i\in\{1,2,\cdots, N\}\nonumber \\
{\phi}_{x,\mathcal{V}}(i)&=&\mathrm{arg}( {x}_a(i)),
\label{ampm}
\end{eqnarray}
where $\mathrm{arg}(\cdot)$ denotes the 4-quadrant arctangent~function which takes values in the range $(-\pi,\pi]$.
\end{definition}
For $\mathbf{A=C}$, (\ref{ampm}) reduces to 1D AM and PM definitions.
This is because setting $\mathbf{A=C}$ results in the graph analytic signal to coincide with the conventional analytic signal as we have discussed in Section \ref{sec:gas}. This in turn implies that the amplitude and phase of the graph analytic signal also coincide with the conventional definitions. We hereafter refer to $\mathcal{A}_{x,\mathcal{V}}$ as the graph AM and ${\phi}_{x,\mathcal{V}}$ as the graph PM. We next discuss computing  frequency-modulation for the graph signal.
\vspace{-.in}

\begin{definition}[Frequency modulation]
The frequency modulation (FM) of a graph signal $\mathbf{x}$ is defined as $\omega_{x,\mathcal{V}}=\phi^u_{x,\mathcal{V}}-\mathbf{A}\phi^u_{x,\mathcal{V}}$,
where $\phi^u_{x,\mathcal{V}}$ denotes the unwrapped phase of the graph analytic signal.
\end{definition}
The unwrapped phase $\phi^u_{x,\mathcal{V}}$ is obtained by performing one-dimensional conventional phase-unwrapping on  $\phi_{x,\mathcal{V}}$ \cite{Oppenheim,Bovik,Tribolet}.
The frequency modulation definition generalizes the backward-difference operator used to compute conventional FM for 1D signals \cite{Oppenheim} defined as the derivative of the phase angle of the analytic signal. 
The phase-unwrapping operation is performed since $\mathrm{arg}(\cdot)$ function returns phase values wrapped in the range $(-\pi,\pi]$ \cite{Arun_amfm2, Oppenheim, Tribolet}. We assume that $\mathbf{A}$ is normalized such that $|\lambda|_{{max}}=1$.
In order to visualize the proposed graph AM and graph FM, we consider speech signal viewed as a graph signal using the linear prediction coefficients as proposed in \cite{ArunEusipco15}. For each speech frame,  we construct $\mathbf{A}$ by connecting every sample to its preceding $P$ samples with edge-weights equal to the corresponding $P$th-order linear prediction coefficients. We plot the obtained graph AM and graph FM in Figure \ref{kabree}. We also include the 1D AM and FM for comparison. We observe that the graph FM is smoother than 1D FM, and the graph AM and 1D AM nearly coincide. 
 \begin{figure}[t]
  {\vspace{-0.in}
  \centering
  \hspace{-.in}
  \subfigure[\hspace{-.2in}]{\hspace{-.001in}
  \includegraphics[width=2.2in]{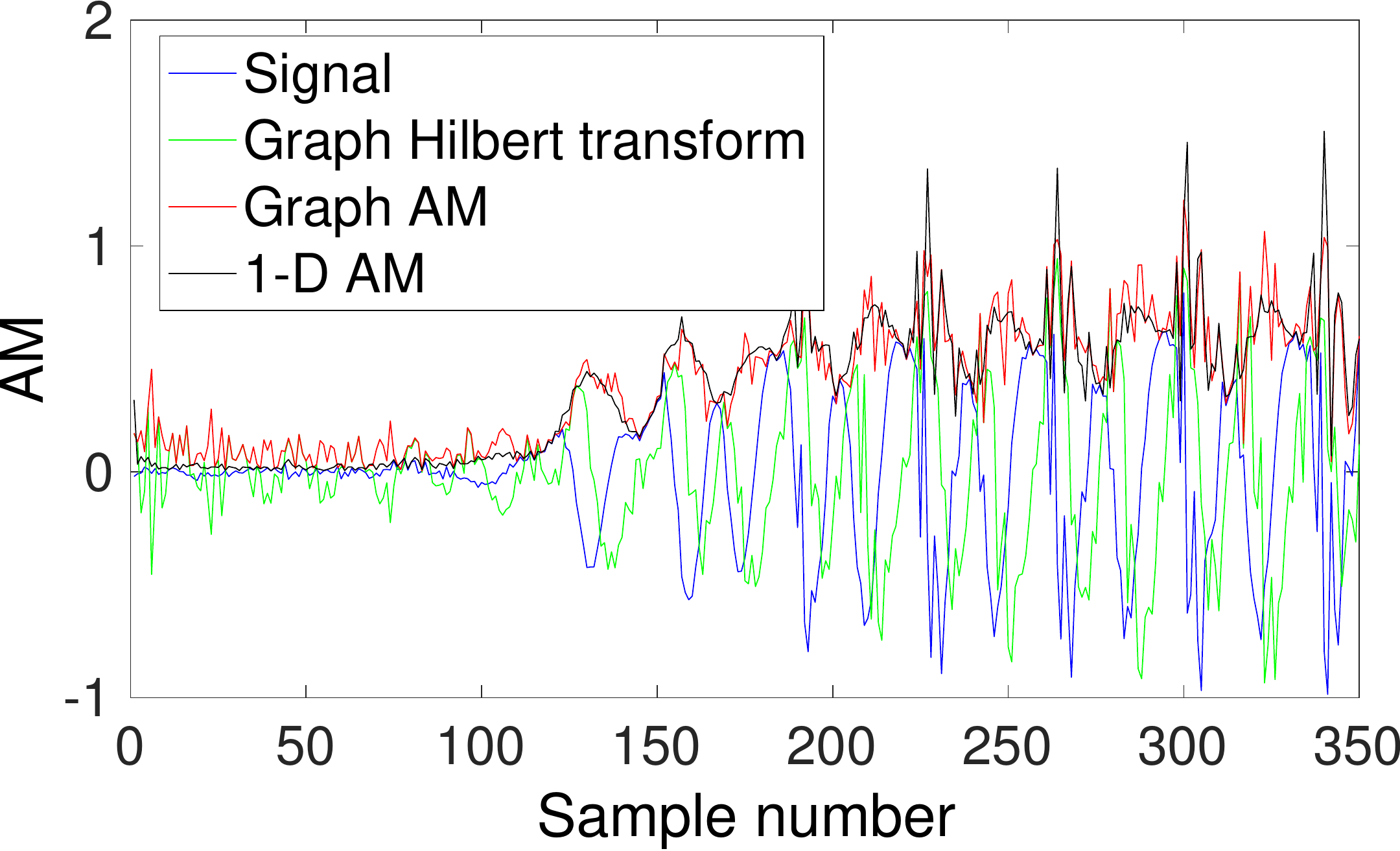}
  } 
  \subfigure[\hspace{.2in}]{\vspace{-.in}
  \includegraphics[width=2.2in]{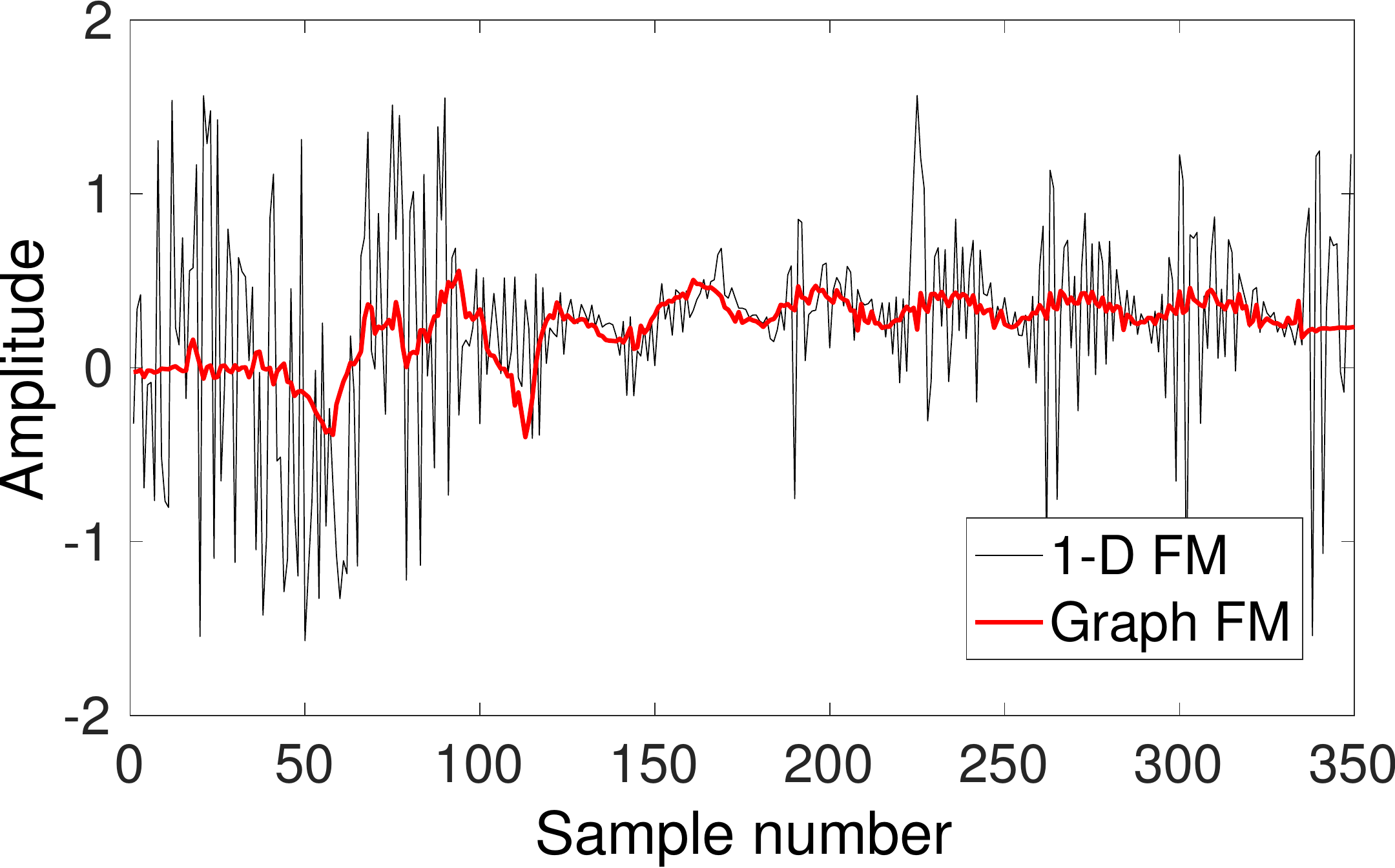}
  }
  \caption{Speech signal, female utterance of the word 'Head', sampled at $16\,$kHz, taken from the NTVD database \cite{NTVDb}. (a) AM, and (b) FM for $P=8$.}
  \label{kabree}}
  \vspace{-0.1in}
\end{figure} 

\section{Experiments}
We next illustrate the applications of the proposed concepts on few synthesized and real-world signal examples. 
\subsection{Graph Hilbert transform and highlighting of singularities/anomalies} 
\label{GHT_edge_sec}
We consider first experiments that demonstrate the edge-highligting behaviour of the graph Hilbert transform in simulated small-world graphs and 2D-image graphs.
The conventional Hilbert transform has been shown to be useful for highlighting singularities in 1D/2D signals \cite{Lohmann, ArunFrHT}. This is a consequence of the functional form of impulse response of the Hilbert transform. Since the graph Hilbert transform generalizes the discrete Hilbert transform, our hypothesis is that the graph Hilbert transform also highlights singularities. We have already seen how the conventional 1D-Hilbert transform, as a special case of the graph Hilbert transform when $\mathbf{A=C}$, highlights edges or anomalies (cf. Figure \ref{GHT_IR}). We next consider a $40\times 40$ 2D signal or image signal. The image is a section of the {\emph{coins}
	image taken from the MATLAB library. Since there is no unique directed graph for an image signal, we define the graph as an extension from the 1D-setting, that is, we consider that $j$th pixel in $i$th row to be connected to the $(j+1)$th pixel in the same row and to the $j$th pixel in the $(i+1)$th row. The corresponding graph then has the adjacency matrix $\mathbf{A}=\mathbf{C}\otimes \mathbf{C}$ as shown in Figure \ref{2Dedge}(a). The image signal and its graph Hilbert transform (reshaped as an image) are shown in Figures \ref{2Dedge}(b) and \ref{2Dedge}(c), respectively. We observe that the graph Hilbert transform specialized to the 2D signal case exhibits edge highlighting behavior. We note here that connecting the pixels differently leads to alternative directed graphs, and we find in our experiments that the corresponding graph Hilbert transforms also highlight edges. However, all these cases are not reported here to avoid repetition.}
\begin{figure}[t]
	\vspace{-.in}
	{\centering
		\hspace{-.in}\subfigure[]{\hspace{-.3in}
			\includegraphics[width=1.6in]{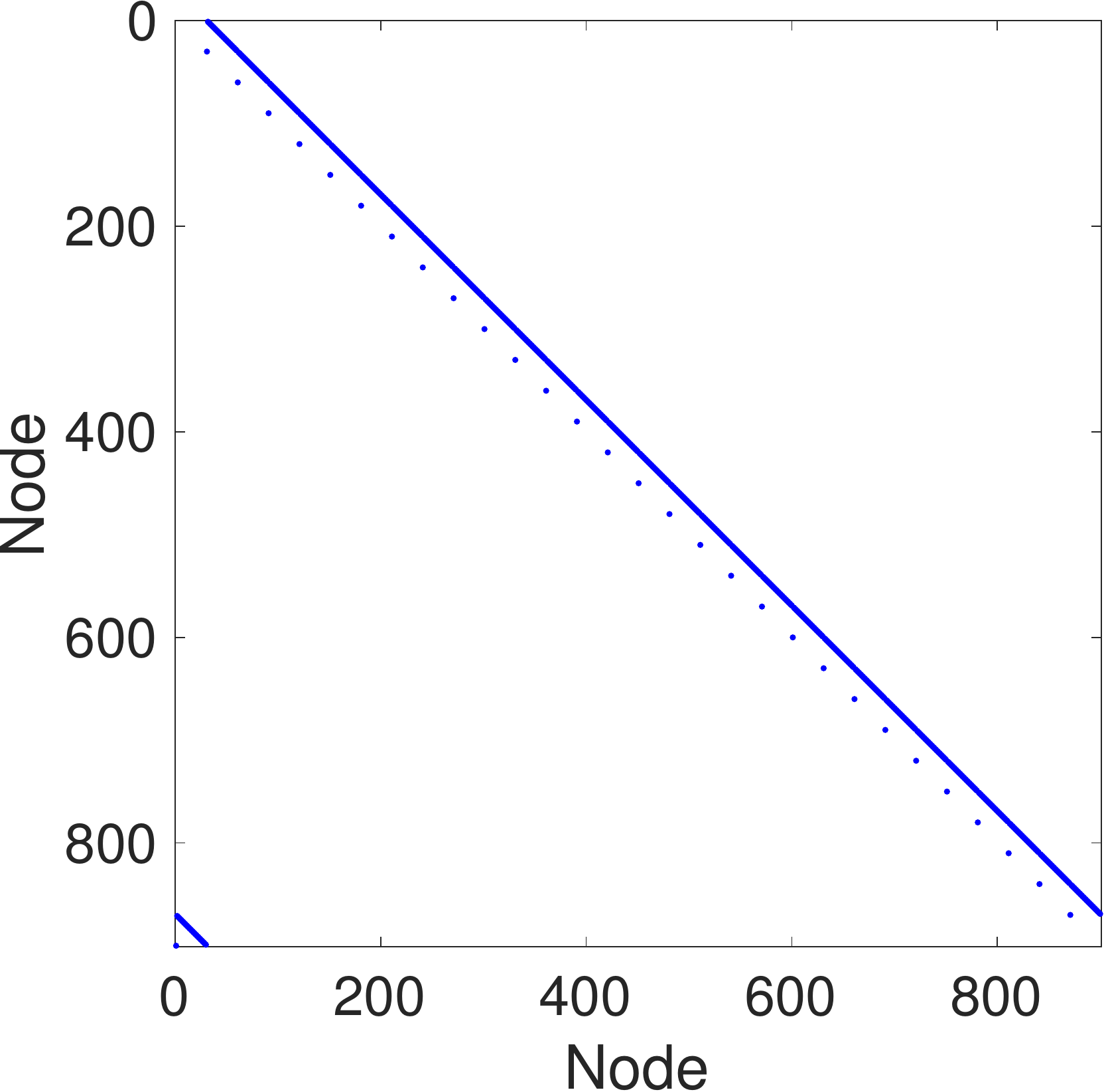}}
		\subfigure[]{
			\includegraphics[width=1.6in]{2Dedge.pdf}}
		\hspace{-.in}\subfigure[\hspace{-.0in}]{
			\includegraphics[width=1.6in]{2DedgeGHT.pdf}}
		\caption{Anomaly highlighting behavior of the graph Hilbert transform for 2D signal graph. (a) $\mathbf{A}$, (b) signal, and (c) graph Hilbert transform.}
		\label{2Dedge}
	}
\end{figure}

We next consider a synthesized social network graph consisting of 10 communities with 6 member nodes each. The nodes within each community are strongly connected in addition to having inter-community edges. 
The intra-community edge-weights are drawn from the uniform distribution over $[0,1]$, and the inter-community edge-weights are drawn from uniform distribution over $[0,0.5]$, and randomly placed across nodes from different communities (Note that the resulting graph is highly assymmetric). Graphs with real edge weights have been extensively employed in analyzing data occuring in many practical applications such as road traffic analysis, brain connectivity \cite{Shuman}. We consider the case of weighted random graphs to demonstrate the potential of our concepts to such application areas. We normalize $\mathbf{A}$ have $|\lambda|_{max}=1$. The nodes are labelled to correspond to the row index of the adjacency matrix. We consider two different cases, one with few inter-community edges ($1\%$ of the total number of possible edges in the graph) and the other with denser edges ($10\%$ percent of the total edges possible in the graph). For each case, we compute the graph Hilbert transform using (\ref{graphHT}) for the graph signal which is zero everywhere except at nodes 18 to 23 (which lie in communities 3 and 4) being active. By intuition, we expect all the nodes connected to these nodes which have value zero (thus making a singularity or anomaly) to be highlighted by the graph Hilbert transform. We observe from Figure \ref{community_edge}(b) that this is indeed the case. The graph Hilbert transform takes large values at nodes 15 and 16 since they are strongly connected to node 18 (cf. Figures \ref{community_edge}(a)-(b)). Similarly, presence of strong edge between nodes 18 to 50 results in node 50 being highlighted by the graph Hilbert transform. Similar arguments can be made for nodes 2, 3, and 55, all of which are highlighted by the graph Hilbert transform. We also note that the extent to which a node is highlighted also varies with the strength of the connecting edge. In the case of dense inter-community edges, we observe that the graph Hilbert transform highlights a large number of nodes since the nodes from 18 to 23 are connected to many nodes (cf. Figures \ref{community_edge}(c)-(d)). 
\begin{figure}[]
	\vspace{-.in}
	{
		$
		\begin{array}{cc}
		\subfigure[\hspace{.0in}]{
			\includegraphics[width=2.5in]{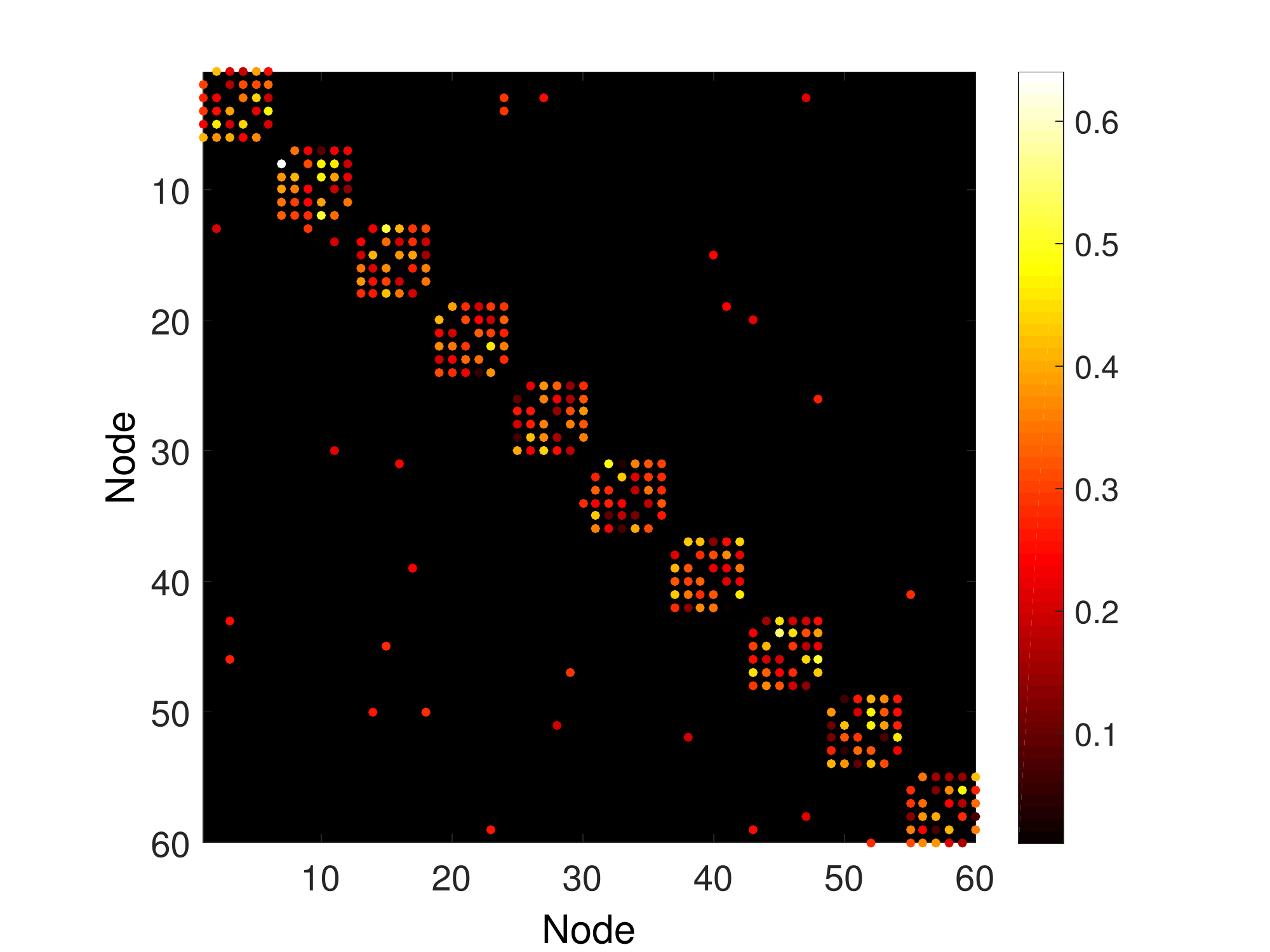}}
		\subfigure[\hspace{-.5in}]{
			\includegraphics[width=2.5in]{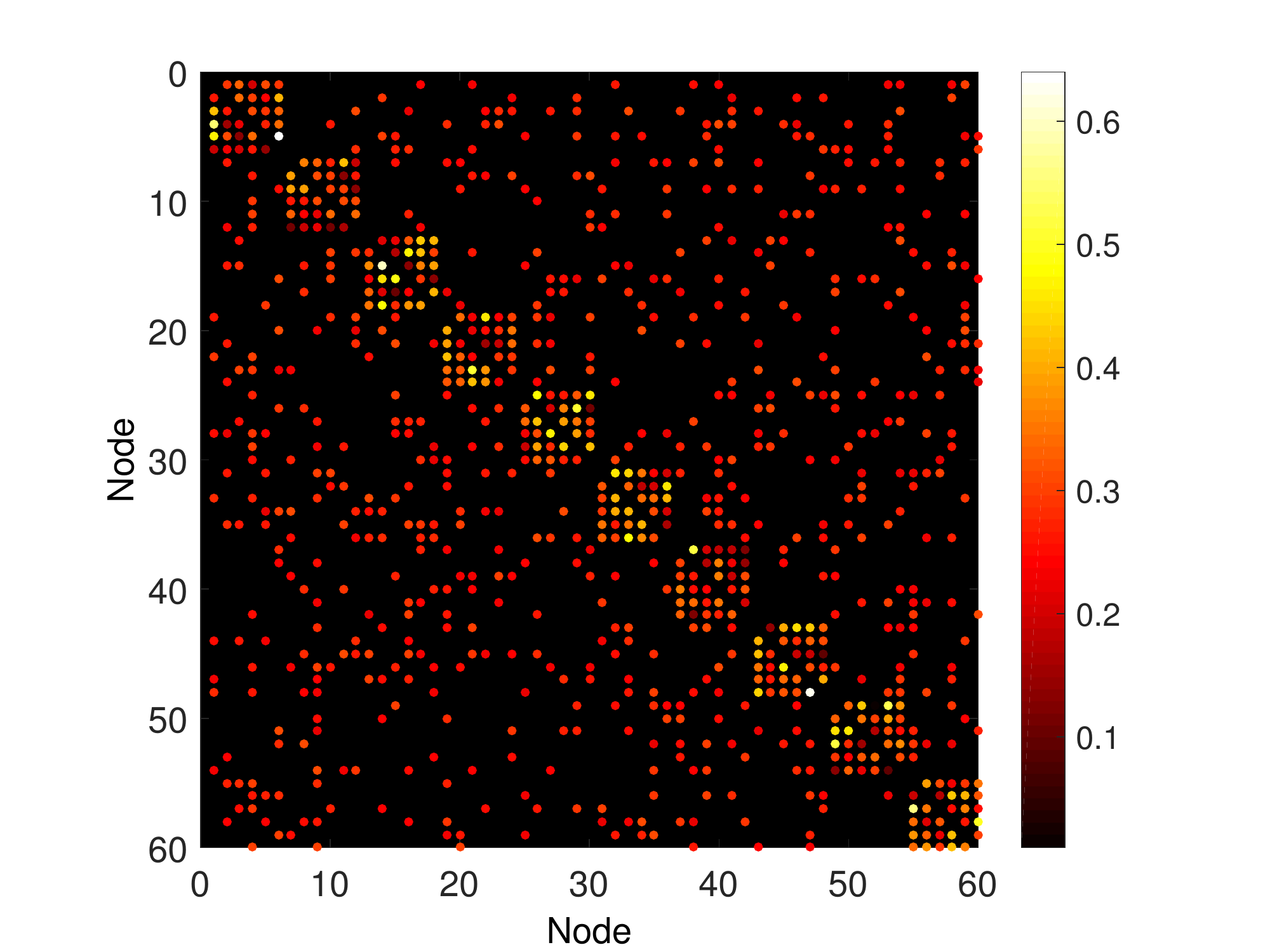}}\\
		\hspace{-.2in}\subfigure[\hspace{-.1in}]{
			\includegraphics[width=2.5in]{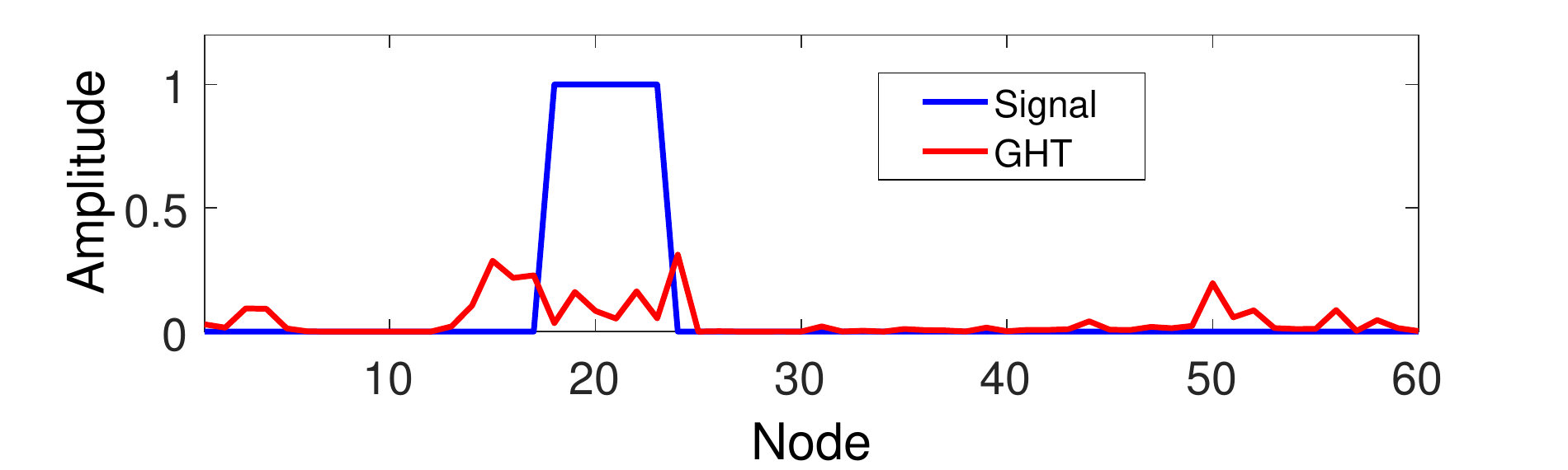}}
		\hspace{-.2in}\subfigure[\hspace{-.1in}]{
			\includegraphics[width=2.5in]{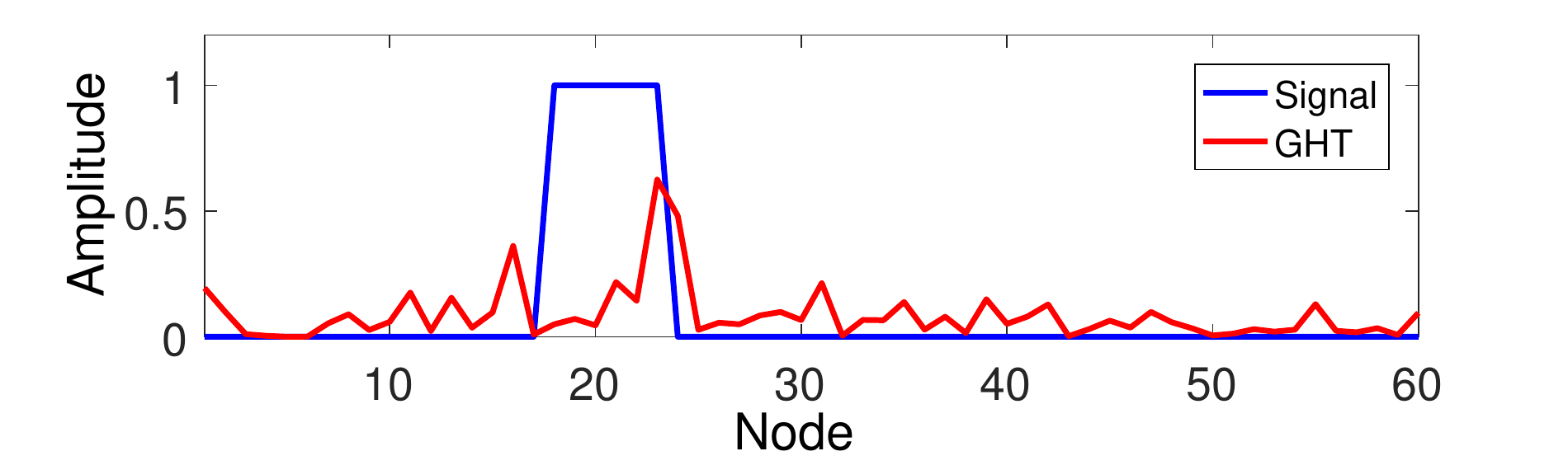}}
		\end{array}$
		\caption{Anomaly highlighting behavior of the graph Hilbert transform for a community graph. (a) $\mathbf{A}$ and (c) Graph Hilbert transform for sparse inter-community connections. (b) $\mathbf{A}$ and (d) graph Hilbert transform for dense inter-community connections.}
		\label{community_edge}
	}
\end{figure}
\begin{figure}[t]
	\vspace{-.15in}
	{\centering
		\subfigure[\hspace{.0in}]{	\hspace{-.5in}
			\includegraphics[width=2.2in]{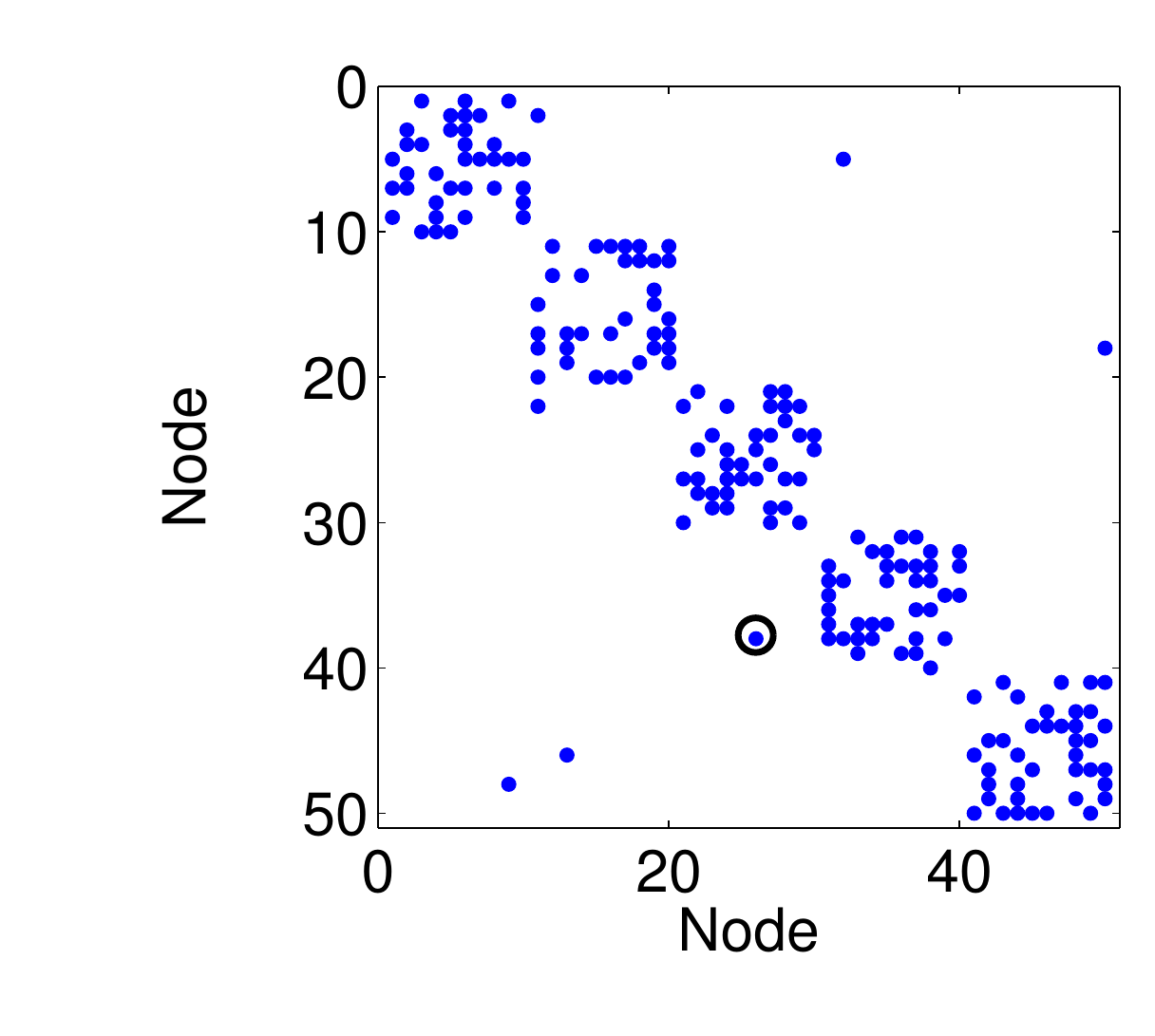}}
		\hspace{-.0in}\subfigure[\hspace{-.1in}]{
			\includegraphics[width=2.in]{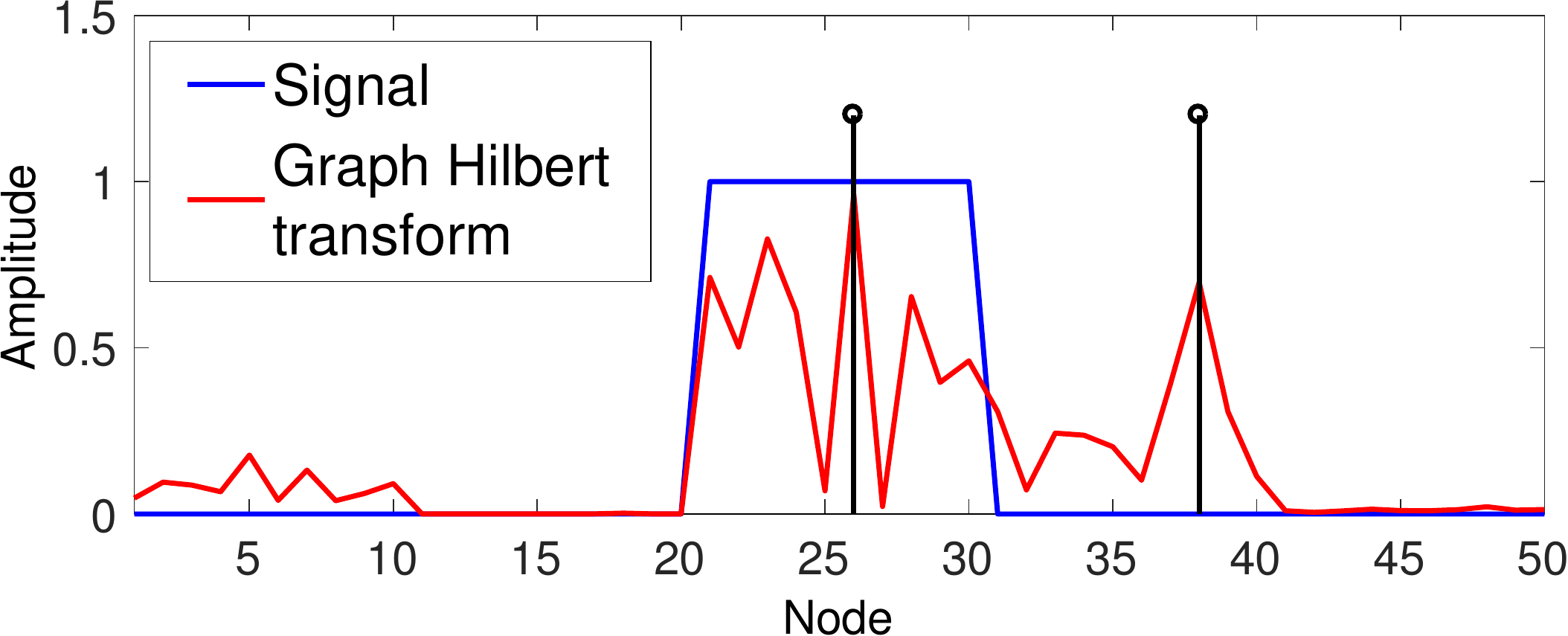}}
		\caption{  Anomaly highlighting behavior of the graph Hilbert transform for unweighted community graph. (a) $\mathbf{A}$ and (b) graph Hilbert transform.
		}
		\label{unweightedGHT}
	}
\end{figure}
In Figure \ref{unweightedGHT}, we consider an unweighted community graph with very few inter community edges. The graph consists of 5 communities of 10 nodes each. Each 10-node community subgraph is randomly generated from the Erd\H{o}s R\'{e}nyi model with an edge probability $p=0.5$. The communities are then connected with very few links also generated randomly. The resulting adjacency matrix is shown in Figure \ref{unweightedGHT}(a). We consider the signal to be all ones corresponding to nodes of community 3. In the present example, community 3 has only one outgoing edge from community 3 (from node 26 to node 38), highlighted by the circle in Figure \ref{unweightedGHT}(a). We observe that the graph Hilbert transform highlights both nodes 38 and 26, as expected, in addition to highlighting the subset of nodes of community 3. 
	
	Our experiments suggest that the graph Hilbert transform could be potentially used in anomaly/edge detection in graphs, particularly because not only the presence but also the location of the anomaly is highlighted.\\

\subsection{Male-female voice classification using graph AM and FM}
 We consider the speech signal as a signal over a graph. Our hypothesis is that viewing the speech signal as a graph signal provides additional information that could help improve the speaker recognition performance. In order to test our hypothesis, we construct a speech graph from learning set data consisting of speech samples from two speakers. We then compute the conventional AM and FM, and the graph AM and FM and use them features for classification. We use two-layer neural network classifiers trained from data distinct from test and learning data. Let $\mathbf{X}_{1}=[\mathbf{x}_{1,1},\cdots,\mathbf{x}_{1,n}]$ and  $\mathbf{X}_{2}=[\mathbf{x}_{2,1},\cdots,\mathbf{x}_{2,n}]$ denote the speech sample matrices from two speakers $S1$ and $S2$ such that $\mathbf{x}_{i,j} \in \mathbb{R}^N$ denotes the $j$th frame of speech samples from $i$th speaker. The speech frames are taken from different sentences uttered by the speakers (one male and other female) from the CMU Arctic Database \cite{CMUArctic}. We choose a frame-length of $N=50$ and total number of frames 4000 (where $n=2000$) such that $\mathbf{X}_l=[\mathbf{X}_1\,\,\mathbf{X}_2]$.
We compute the adjacency matrix by solving the following optimization problem:
\begin{align}
\mathbf{A}^*&=\arg\min_{\mathbf{A}}\|\mathbf{X}_l-\mathbf{AX}_l\|_2^2\nonumber\\
\mbox{subject to}&\quad \mbox{diag}(\mathbf{A})=\mathbf{0},\,\,\mathbf{A 1}=\mathbf{1},\,\, \mathbf{A}^T \mathbf{1}=\mathbf{1}.\qquad
\label{learnA}
\end{align}
We use the constraints $\mathbf{A 1}=\mathbf{1},\quad \mathbf{A}^T \mathbf{1}=\mathbf{1}$ to avoid ill-conditioning in case of insufficient learning data.
In Figure \ref{learnt_graph}, we plot the adjacency matrix obtained from (\ref{learnA}). 
  \begin{figure}[t]
  \vspace{-.15in}
  {\centering
  $
  \begin{array}{cc}
  \hspace{-.0in}\hspace{-.in}
  \subfigure[]{\includegraphics[width=2.in]{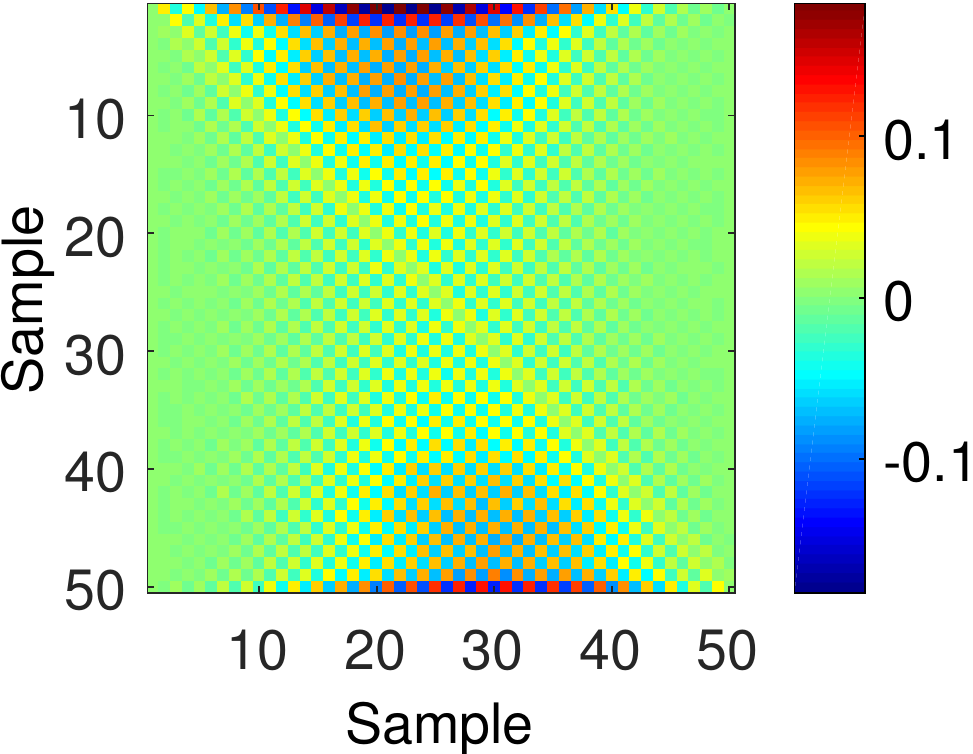}
  }
  \subfigure[]{
  \includegraphics[width=1.8in]{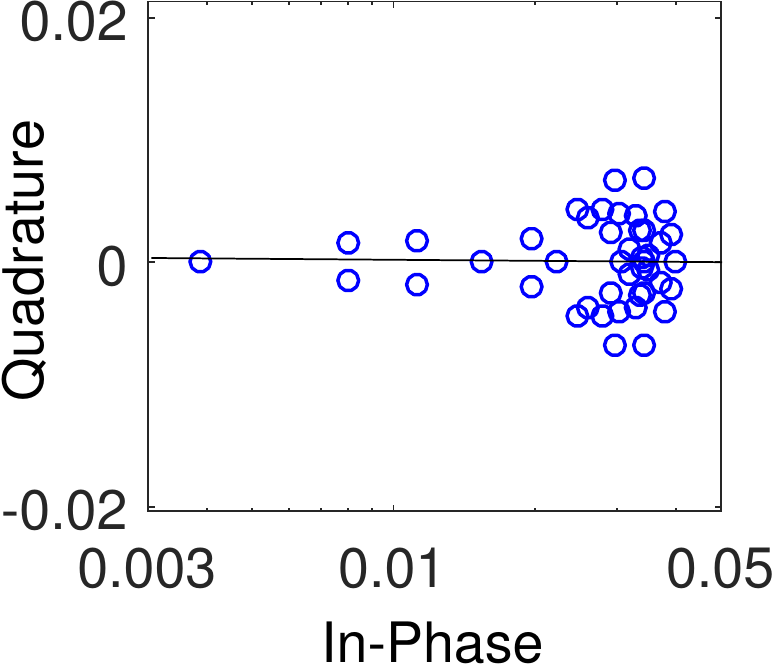}
  }\\
  \hspace{-.0in}
  \subfigure[]{
    \includegraphics[width=2.in]{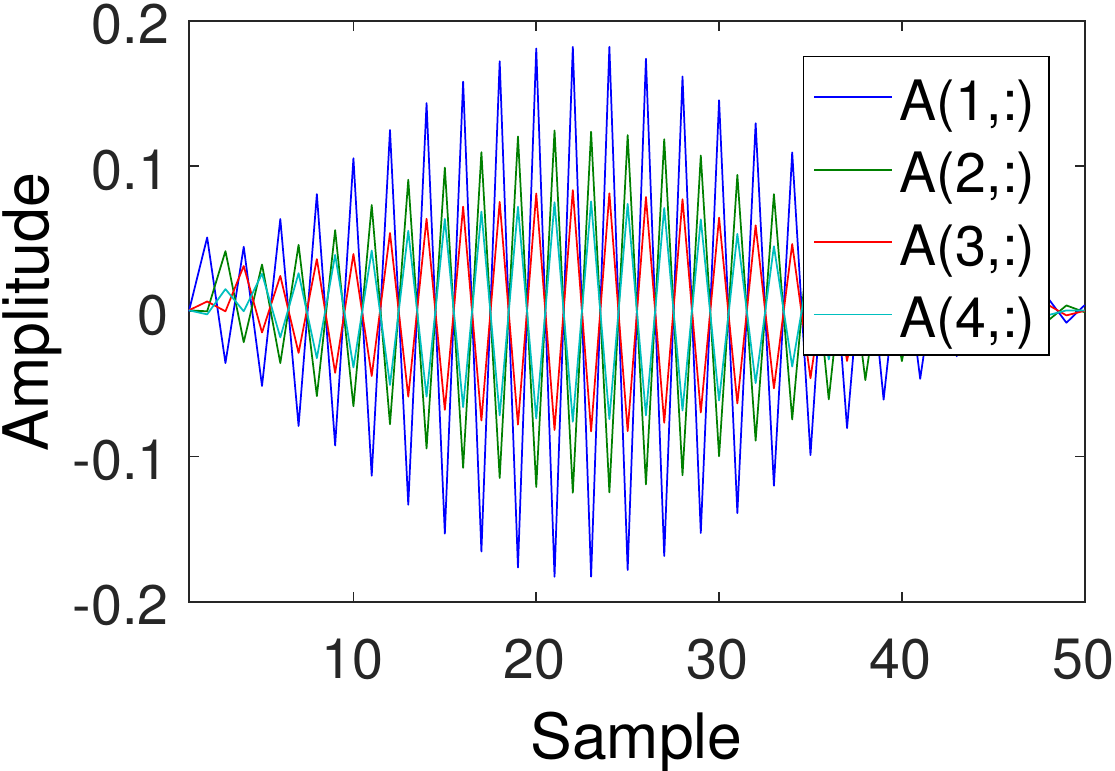}
    }
    \subfigure[]{\includegraphics[width=2.in]{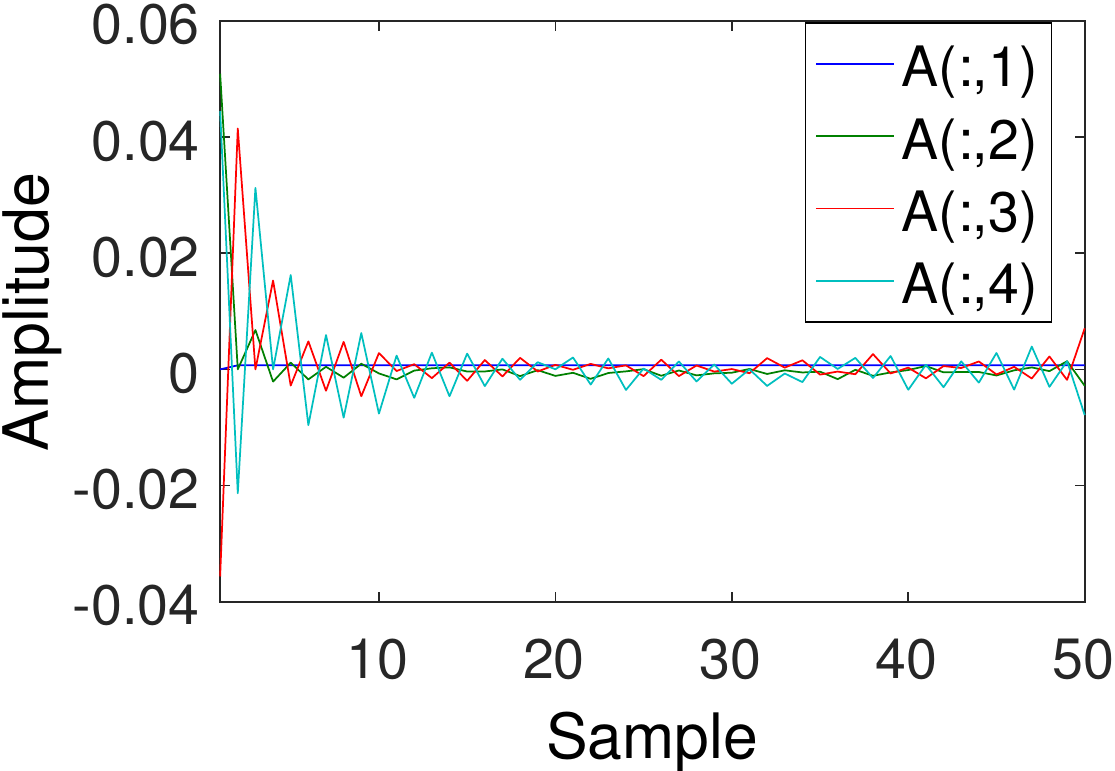}
    }
  \end{array}$ 
    \caption{Graph learnt from speech samples. (a) $\mathbf{A}$, (b) eigenvalues of $\mathbf{A}$, (c) first four rows of $\mathbf{A}$, and (d) first four columns of $\mathbf{A}$.}
  \label{learnt_graph}
  \vspace{-0.in}
  }
\end{figure}
We consider three different classifiers: \newline
 \textbf{Classifier 1}: The classifier uses magnitudes of the DFT coefficients of 1D AM and FM as feature vectors. The feature vector has length 100.
 \newline
 \textbf{Classifier 2}: The classifier uses the magnitudes of the GFT of the graph AM and FM as feature vectors, where the GFT is obtained from the eigen-decomposition of $\mathbf{A}^*$. The feature vector has length 100.
 \newline
 \textbf{Classifier 3}:  The classifier uses the magnitudes of the DFT coefficients of 1D AM and FM, concatenated with magnitudes of the GFT of the graph AM and FM as feature vectors. The length of the feature vector is equal to 200.
\\\newline
The classifiers are trained using the features from training data $\mathbf{X}_{tr}$ and tested on $\mathbf{X}_{test}$, both data sets being different from $\mathbf{X}_l$ used in computing the adjacency matrix. The composite dataset $[\mathbf{X}_{tr} \mathbf{X}_{test}]$ consists of $5\times 10^4$ samples of which 60\% is $\mathbf{X}_{tr}$ and the rest in $\mathbf{X}_{test}$. The classifier performance is computed for different  number of sigmoidal neurons in the hidden layer. The performance is averaged over 50 runs where the training and test data are randomly partitioned. We observe from Table \ref{speaker_classifier_mean} that Classifier 3 outperforms the other classifiers with a classification improvement of up to 2\% in comparison with the DFT-based classifier. We also observe that the performance of the neural classifier saturates after 5 hidden neurons. This shows that the proposed graph amplitude and frequency modulations improve speaker classification performance, and that viewing speech as a graph signal indeed provides complementary information. 
    \begin{table}
    \centering
  \begin{tabular}{|c|c|c|c|}
    \hline
     Number of & Classifier 1& Classifier 2& Classifier 3 \\
     hidden neurons&  (DFT) & (GFT) & (DFT+GFT) \\
    \hline
    1 & 69.4 & 59.6  &  \bf{69.6}\\
    5  &   70.8  &  59.7 &\bf{72.6}\\
    10&  71.0   &     59.3 & \bf{72.2}\\
      \hline
\end{tabular}
  \hspace{.0in}
  \caption{Classification accuracy (in percentage)  obtained over 50 runs of the neural network.}
  \label{speaker_classifier_mean}
  \vspace{-.2in}
\end{table}
 %
\section{Discussions and Conclusions}
We proposed definitions for the analytic signal and Hilbert transform of real graph signals over directed graphs. 
We showed that graph Hilbert transform and graph analytic signal are linear and shift-invariant over graphs, and that they inherit many properties, and in particular anomaly/singularity highlighting property for some graphs of interest. We also demonstrated through a numerical example that the graph Hilbert transform does not inherit the Bedrosian property in general. Using the graph analytic signal, we defined amplitude, phase, and frequency modulations for graph signals. 
In order to illustrate the proposed notions, we considered their application to synthesized and real-world signal examples. We observed in the context of speaker recognition that viewing the speech signal as a graph signal resulted in improved classification performance. This is because the graph signal model captures signal correlation across all samples in a speech frame, unlike the 1D graph which considers only the preceding sample. 
The ability of the graph Hilbert transform to highlight edge/singularities/anomalies could be potentially employed in analyzing scenarios such as malfunctioning in power grids, spread of disease or epidemics, identifying activity sources in the brain, and traffic bottlenecks over transportation networks, and study of outliers in social network trends. We note that applications chosen in this article serve the purpose of illustrating the proposed concepts and are by no means exhaustive. As with 1D modulation analysis, the utility varies across applications and can only be revealed by detailed analysis on various datasets. 


\begin{thebibliography}{9}
	\expandafter\ifx\csname natexlab\endcsname\relax\def\natexlab#1{#1}\fi
	\providecommand{\bibinfo}[2]{#2}
	\ifx\xfnm\relax \def\xfnm[#1]{\unskip,\space#1}\fi
	\bibitem[{Newman(2010)}]{Newman}
	\bibinfo{author}{M.~E.~J. Newman}, \bibinfo{title}{Networks: An Introduction},
	\bibinfo{publisher}{Oxford University Press}, \bibinfo{year}{2010}.
	\bibitem[{Sandryhaila and Moura(2014)}]{Sandry3}
	\bibinfo{author}{A.~Sandryhaila}, \bibinfo{author}{J.~M.~F. Moura},
	\newblock \bibinfo{title}{Discrete signal processing on graphs: Frequency
		analysis},
	\newblock \bibinfo{journal}{IEEE Trans. Signal Process.} \bibinfo{volume}{62}
	(\bibinfo{year}{2014}) \bibinfo{pages}{3042--3054}.
	\bibitem[{Shuman et~al.(2013)Shuman, Narang, Frossard, Ortega, and
		Vandergheynst}]{Shuman}
	\bibinfo{author}{D.~I. Shuman}, \bibinfo{author}{S.~Narang},
	\bibinfo{author}{P.~Frossard}, \bibinfo{author}{A.~Ortega},
	\bibinfo{author}{P.~Vandergheynst},
	\newblock \bibinfo{title}{The emerging field of signal processing on graphs:
		Extending high-dimensional data analysis to networks and other irregular
		domains},
	\newblock \bibinfo{journal}{IEEE Signal Process. Mag.} \bibinfo{volume}{30}
	(\bibinfo{year}{2013}) \bibinfo{pages}{83--98}.
	\bibitem[{Teke and Vaidyanathan(2017{\natexlab{b}})}]{pp_graph2}
	\bibinfo{author}{O.~Teke}, \bibinfo{author}{P.~P. Vaidyanathan},
	\newblock \bibinfo{title}{Extending classical multirate signal processing
		theory to graphs-part ii: M-channel filter banks},
	\newblock \bibinfo{journal}{IEEE Trans. Signal Process.} \bibinfo{volume}{65}
	(\bibinfo{year}{2017}{\natexlab{b}}) \bibinfo{pages}{423--437}.
	\bibitem[{Raif and Guibas(2013)}]{deeplearninggraphwavelets}
	\bibinfo{author}{R.~Raif}, \bibinfo{author}{L.~Guibas},
	\newblock \bibinfo{title}{Wavelets on graphs via deep learning},
	\newblock \bibinfo{journal}{Proc. Adv. Neural Inform. Proces. Syst.}
	(\bibinfo{year}{2013}) \bibinfo{pages}{998--1006}.
	\bibitem[{Jansen et~al.(2009)Jansen, Nason, and Silverman}]{multiscalegraphs}
	\bibinfo{author}{M.~Jansen}, \bibinfo{author}{G.~P. Nason},
	\bibinfo{author}{B.~W. Silverman},
	\newblock \bibinfo{title}{Multiscale methods for data on graphs and irregular
		multidimensional situations},
	\newblock \bibinfo{journal}{J. Roy. Statist. Soc.: Series B}
	\bibinfo{volume}{71} (\bibinfo{year}{2009}) \bibinfo{pages}{97--125}.
	\bibitem[{Sandryhaila and Moura(2013)}]{Sandry1}
	\bibinfo{author}{A.~Sandryhaila}, \bibinfo{author}{J.~M.~F. Moura},
	\newblock \bibinfo{title}{Discrete signal processing on graphs},
	\newblock \bibinfo{journal}{IEEE Trans. Signal Process.} \bibinfo{volume}{61}
	(\bibinfo{year}{2013}) \bibinfo{pages}{1644--1656}.
	\bibitem[{Sandryhaila and Moura(2014)}]{Sandry2}
	\bibinfo{author}{A.~Sandryhaila}, \bibinfo{author}{J.~M.~F. Moura},
	\newblock \bibinfo{title}{Big data analysis with signal processing on graphs:
		Representation and processing of massive data sets with irregular structure},
	\newblock \bibinfo{journal}{IEEE Signal Process. Mag.} \bibinfo{volume}{31}
	(\bibinfo{year}{2014}) \bibinfo{pages}{80--90}.
	\bibitem[{Cohen(1995)}]{Cohen}
	\bibinfo{author}{L.~Cohen}, \bibinfo{title}{{T}ime-{F}requency {A}nalysis:
		{T}heory and {A}pplications}, \bibinfo{publisher}{Prentice-Hall, Inc., Upper
		Saddle River, NJ, USA}, \bibinfo{edition}{1} edition, \bibinfo{year}{1995}.
	\bibitem[{Cusmariu(2002)}]{Cusmariu}
	\bibinfo{author}{A.~Cusmariu},
	\newblock \bibinfo{title}{Fractional analytic signals},
	\newblock \bibinfo{journal}{Signal Process.} \bibinfo{volume}{82}
	(\bibinfo{year}{2002}) \bibinfo{pages}{267--~272}.
	\bibitem[{Zayed(1998)}]{Zayed}
	\bibinfo{author}{A.~I. Zayed},
	\newblock \bibinfo{title}{Hilbert transform associated with the fractional
		{F}ourier transform},
	\newblock \bibinfo{journal}{IEEE Signal Process. Lett.} \bibinfo{volume}{5}
	(\bibinfo{year}{1998}) \bibinfo{pages}{206--208}.
	\bibitem[{Guanlei et~al.(2009)Guanlei, Xiaotong, and Xiaogang}]{Guanlei}
	\bibinfo{author}{X.~Guanlei}, \bibinfo{author}{W.~Xiaotong},
	\bibinfo{author}{X.~Xiaogang},
	\newblock \bibinfo{title}{Generalized {H}ilbert transform and its properties in
		{2D} {LCT} domain},
	\newblock \bibinfo{journal}{Signal Process.} \bibinfo{volume}{89}
	(\bibinfo{year}{2009}) \bibinfo{pages}{1395--1402}.
	\bibitem[{Sarkar et~al.(2009)Sarkar, Mukherjee, and Ray}]{Sarkar}
	\bibinfo{author}{S.~Sarkar}, \bibinfo{author}{K.~Mukherjee},
	\bibinfo{author}{A.~Ray},
	\newblock \bibinfo{title}{Generalization of {H}ilbert transform for symbolic
		analysis of noisy signals},
	\newblock \bibinfo{journal}{Signal Process.} \bibinfo{volume}{89}
	(\bibinfo{year}{2009}) \bibinfo{pages}{1245--1251}.
	\bibitem[{Zhang and Li(2018)}]{Phi_AS}
	\bibinfo{author}{Y.-N. Zhang}, \bibinfo{author}{B.-Z. Li},
	\newblock \bibinfo{title}{$\phi$-linear canonical analytic signals},
	\newblock \bibinfo{journal}{Signal Process.} \bibinfo{volume}{143}
	(\bibinfo{year}{2018}) \bibinfo{pages}{181 -- 190}.
	\bibitem[{Felsberg and Sommer(2001)}]{monogenic}
	\bibinfo{author}{M.~Felsberg}, \bibinfo{author}{G.~Sommer},
	\newblock \bibinfo{title}{The monogenic signal},
	\newblock \bibinfo{journal}{IEEE Trans. Signal Proc.} \bibinfo{volume}{49}
	(\bibinfo{year}{2001}) \bibinfo{pages}{3136--3144}.
	\bibitem[{B\"{u}low and Sommer(1999)}]{2D_AS_1}
	\bibinfo{author}{T.~B\"{u}low}, \bibinfo{author}{G.~Sommer},
	\newblock \bibinfo{title}{A novel approach to the {2D} analytic signal},
	\newblock in: \bibinfo{editor}{F.~Solina}, \bibinfo{editor}{A.~Leonardis}
	(Eds.), \bibinfo{booktitle}{Computer Analysis of Images and Patterns}, volume
	\bibinfo{volume}{1689} of \textit{\bibinfo{series}{Lecture Notes in Computer
			Science}}, \bibinfo{publisher}{Springer Berlin Heidelberg},
	\bibinfo{year}{1999}, pp. \bibinfo{pages}{25--32}.
	\bibitem[{Bernstein(2014)}]{Bernstein2014}
	\bibinfo{author}{S.~Bernstein}, \bibinfo{title}{The Fractional Monogenic
		Signal}, \bibinfo{publisher}{Springer International Publishing}, pp.
	\bibinfo{pages}{75--88}.
	\bibitem[{Venkitaraman et~al.(2015)Venkitaraman, Chatterjee, and
		Handel}]{ArunEusipco15}
	\bibinfo{author}{A.~Venkitaraman}, \bibinfo{author}{S.~Chatterjee},
	\bibinfo{author}{P.~Handel},
	\newblock \bibinfo{title}{Graph linear prediction results in smaller error than
		standard linear prediction},
	\newblock \bibinfo{journal}{Proc. Eur. Signal Process. Conf. (EUSIPCO)}
	(\bibinfo{year}{2015}).
	\bibitem[{Shuman et~al.(2012)Shuman, Ricaud, and Vandergheynst}]{windowedGFT}
	\bibinfo{author}{D.~I. Shuman}, \bibinfo{author}{B.~Ricaud},
	\bibinfo{author}{P.~Vandergheynst},
	\newblock \bibinfo{title}{A windowed graph {F}ourier transform},
	\newblock \bibinfo{journal}{{IEEE} Statist. Signal Process. Workshop (SSP)}
	(\bibinfo{year}{2012}) \bibinfo{pages}{133--136}.
	\bibitem[{Narang and Ortega(2013)}]{Narang2013}
	\bibinfo{author}{S.~K. Narang}, \bibinfo{author}{A.~Ortega},
	\newblock \bibinfo{title}{Compact support biorthogonal wavelet filterbanks for
		arbitrary undirected graphs},
	\newblock \bibinfo{journal}{{IEEE} Trans. Signal Process.} \bibinfo{volume}{61}
	(\bibinfo{year}{2013}) \bibinfo{pages}{4673--4685}.
	\bibitem[{Narang and Ortega(2012)}]{Narang2012}
	\bibinfo{author}{S.~K. Narang}, \bibinfo{author}{A.~Ortega},
	\newblock \bibinfo{title}{Perfect reconstruction two-channel wavelet filter
		banks for graph structured data},
	\newblock \bibinfo{journal}{IEEE Trans. Signal Process.} \bibinfo{volume}{60}
	(\bibinfo{year}{2012}) \bibinfo{pages}{2786--2799}.
	\bibitem[{Tay et~al.(2017)Tay, Tanaka, and Sakiyama}]{TAY201766}
	\bibinfo{author}{D.~B. Tay}, \bibinfo{author}{Y.~Tanaka},
	\bibinfo{author}{A.~Sakiyama},
	\newblock \bibinfo{title}{Critically sampled graph filter banks with polynomial
		filters from regular domain filter banks},
	\newblock \bibinfo{journal}{Signal Process.} \bibinfo{volume}{131}
	(\bibinfo{year}{2017}) \bibinfo{pages}{66 -- 72}.
	\bibitem[{Coifman and Maggioni(2006)}]{Coifman2006}
	\bibinfo{author}{R.~R. Coifman}, \bibinfo{author}{M.~Maggioni},
	\newblock \bibinfo{title}{Diffusion wavelets},
	\newblock \bibinfo{journal}{Appl. Comput. Harmonic Anal.} \bibinfo{volume}{21}
	(\bibinfo{year}{2006}) \bibinfo{pages}{53--94}.
	\bibitem[{Ganesan et~al.(2005)Ganesan, Greenstein, Estrin, Heidemann, and
		Govindan}]{Ganesan}
	\bibinfo{author}{D.~Ganesan}, \bibinfo{author}{B.~Greenstein},
	\bibinfo{author}{D.~Estrin}, \bibinfo{author}{J.~Heidemann},
	\bibinfo{author}{R.~Govindan},
	\newblock \bibinfo{title}{Multiresolution storage and search in sensor
		networks},
	\newblock \bibinfo{journal}{ACM Trans. Storage} \bibinfo{volume}{1}
	(\bibinfo{year}{2005}) \bibinfo{pages}{277--315}.
	\bibitem[{Hammond et~al.(2011)Hammond, Vandergheynst, and
		Gribonval}]{Hammond2011}
	\bibinfo{author}{D.~K. Hammond}, \bibinfo{author}{P.~Vandergheynst},
	\bibinfo{author}{R.~Gribonval},
	\newblock \bibinfo{title}{Wavelets on graphs via spectral graph theory},
	\newblock \bibinfo{journal}{Appl. Comput. Harmonic Anal.} \bibinfo{volume}{30}
	(\bibinfo{year}{2011}) \bibinfo{pages}{129--150}.
	\bibitem[{Wagner et~al.(2006)Wagner, Delouille, and Baraniuk}]{Wagner2}
	\bibinfo{author}{R.~Wagner}, \bibinfo{author}{V.~Delouille},
	\bibinfo{author}{R.~Baraniuk},
	\newblock \bibinfo{title}{Distributed wavelet de-noising for sensor networks},
	\newblock \bibinfo{journal}{Proc. 45th IEEE Conf. Decision Control}
	(\bibinfo{year}{2006}) \bibinfo{pages}{373--379}.
	\bibitem[{Nakahira and Miyamoto(2016)}]{parseval_wavelets_gsp}
	\bibinfo{author}{K.~Nakahira}, \bibinfo{author}{A.~Miyamoto},
	\newblock \bibinfo{title}{Parseval wavelets on hierarchical graphs},
	\newblock \bibinfo{journal}{Appl. Comput. Harmonic Anal.}
	(\bibinfo{year}{2016}).
	\bibitem[{Gama et~al.(2016)Gama, Marques, Mateos, and Ribeiro}]{graphsamp1}
	\bibinfo{author}{F.~Gama}, \bibinfo{author}{A.~G. Marques},
	\bibinfo{author}{G.~Mateos}, \bibinfo{author}{A.~Ribeiro},
	\newblock \bibinfo{title}{Rethinking sketching as sampling: Linear transforms
		of graph signals},
	\newblock \bibinfo{journal}{Proc. Asilomar Conf. Signals, Systems Comput.}
	(\bibinfo{year}{2016}) \bibinfo{pages}{522--526}.
	\bibitem[{Chen et~al.(2016)Chen, Varma, Singh, and Kova{\v
			c}evi{\'c}}]{graphsamp3}
	\bibinfo{author}{S.~Chen}, \bibinfo{author}{R.~Varma},
	\bibinfo{author}{A.~Singh}, \bibinfo{author}{J.~Kova{\v c}evi{\'c}},
	\newblock \bibinfo{title}{Signal recovery on graphs: Fundamental limits of
		sampling strategies},
	\newblock \bibinfo{journal}{IEEE Trans. Signal Inform. Process. over Networks}
	\bibinfo{volume}{2} (\bibinfo{year}{2016}) \bibinfo{pages}{539--554}.
	\bibitem[{Tsitsvero et~al.(2016)Tsitsvero, Barbarossa, and
		Lorenzo}]{graphsamp6}
	\bibinfo{author}{M.~Tsitsvero}, \bibinfo{author}{S.~Barbarossa},
	\bibinfo{author}{P.~D. Lorenzo},
	\newblock \bibinfo{title}{Signals on graphs: Uncertainty principle and
		sampling},
	\newblock \bibinfo{journal}{{IEEE} Trans. Signal Process.}
	\bibinfo{volume}{64} (\bibinfo{year}{2016}) \bibinfo{pages}{4845--4860}.
	\bibitem[{Anis et~al.(2016)Anis, Gadde, and Ortega}]{anis}
	\bibinfo{author}{A.~Anis}, \bibinfo{author}{A.~Gadde},
	\bibinfo{author}{A.~Ortega},
	\newblock \bibinfo{title}{Efficient sampling set selection for bandlimited
		graph signals using graph spectral proxies},
	\newblock \bibinfo{journal}{{IEEE} Trans. Signal Process.} \bibinfo{volume}{64}
	(\bibinfo{year}{2016}) \bibinfo{pages}{3775--3789}.
	\bibitem[{Chen et~al.(2015)Chen, Varma, Sandryhaila, and Kovacevic}]{chensamp}
	\bibinfo{author}{S.~Chen}, \bibinfo{author}{R.~Varma},
	\bibinfo{author}{A.~Sandryhaila}, \bibinfo{author}{J.~Kovacevic},
	\newblock \bibinfo{title}{Discrete signal processing on graphs: Sampling
		theory},
	\newblock \bibinfo{journal}{{IEEE} Trans. Signal Process.} \bibinfo{volume}{63}
	(\bibinfo{year}{2015}) \bibinfo{pages}{6510--6523}.
	\bibitem[{Wang et~al.(2016)Wang, Chen, and Gu}]{WANG2016119}
	\bibinfo{author}{X.~Wang}, \bibinfo{author}{J.~Chen}, \bibinfo{author}{Y.~Gu},
	\newblock \bibinfo{title}{Local measurement and reconstruction for noisy
		bandlimited graph signals},
	\newblock \bibinfo{journal}{Signal Process.} \bibinfo{volume}{129}
	(\bibinfo{year}{2016}) \bibinfo{pages}{119 -- 129}.
	\bibitem[{Puy et~al.(2016)Puy, Tremblay, Gribonval, and
		Vandergheynst}]{PUY2016}
	\bibinfo{author}{G.~Puy}, \bibinfo{author}{N.~Tremblay},
	\bibinfo{author}{R.~Gribonval}, \bibinfo{author}{P.~Vandergheynst},
	\newblock \bibinfo{title}{Random sampling of bandlimited signals on graphs},
	\newblock \bibinfo{journal}{Appl. Computat. Harmonic Anal.}
	(\bibinfo{year}{2016}).
	\bibitem[{Girault(2015)}]{statgraph1}
	\bibinfo{author}{B.~Girault},
	\newblock \bibinfo{title}{Stationary graph signals using an isometric graph
		translation},
	\newblock \bibinfo{journal}{Proc. Eur. Signal Process. Conf. (EUSIPCO)}
	(\bibinfo{year}{2015}) \bibinfo{pages}{1516--1520}.
	\bibitem[{Segarra et~al.(2016)Segarra, Marques, Leus, and Ribeiro}]{statgraph2}
	\bibinfo{author}{S.~Segarra}, \bibinfo{author}{A.~G. Marques},
	\bibinfo{author}{G.~Leus}, \bibinfo{author}{A.~Ribeiro},
	\newblock \bibinfo{title}{Stationary graph processes: Nonparametric spectral
		estimation},
	\newblock \bibinfo{journal}{Proc. IEEE Sensor Array Multichannel Signal
		Process. Workshop (SAM)}  (\bibinfo{year}{2016}) \bibinfo{pages}{1--5}.
	\bibitem[{B.~Girault and Fleury(2015)}]{statgraph4}
	\bibinfo{author}{B.~Girault}, \bibinfo{author}{E.~Fleury},
	\newblock \bibinfo{title}{Translation on graphs: An isometric shift operator},
	\newblock \bibinfo{journal}{Signal Process. Lett.} \bibinfo{volume}{22}
	(\bibinfo{year}{2015}) \bibinfo{pages}{2416--2420}.
	\bibitem[{Thanou et~al.(2014)Thanou, {D. I Shuman}, and Frossard}]{Thanou2014}
	\bibinfo{author}{D.~Thanou}, \bibinfo{author}{{D. I Shuman}},
	\bibinfo{author}{P.~Frossard},
	\newblock \bibinfo{title}{Learning parametric dictionaries for signals on
		graphs},
	\newblock \bibinfo{journal}{IEEE Trans. Signal Process.} \bibinfo{volume}{62}
	(\bibinfo{year}{2014}) \bibinfo{pages}{3849--3862}.
	\bibitem[{Shuman et~al.(2016)Shuman, Ricaud, and
		Vandergheynst}]{vertexfrequency}
	\bibinfo{author}{D.~I. Shuman}, \bibinfo{author}{B.~Ricaud},
	\bibinfo{author}{P.~Vandergheynst},
	\newblock \bibinfo{title}{Vertex-frequency analysis on graphs},
	\newblock \bibinfo{journal}{Appl. Comput. Harmonic Anal.} \bibinfo{volume}{40}
	(\bibinfo{year}{2016}) \bibinfo{pages}{260 -- 291}.
	\bibitem[{Shahid et~al.(2015)Shahid, Kalofolias, Bresson, Bronstein, and
		Vandergheynst}]{graphPCA1}
	\bibinfo{author}{N.~Shahid}, \bibinfo{author}{V.~Kalofolias},
	\bibinfo{author}{X.~Bresson}, \bibinfo{author}{M.~Bronstein},
	\bibinfo{author}{P.~Vandergheynst},
	\newblock \bibinfo{title}{Robust principal component analysis on graphs},
	\newblock \bibinfo{journal}{IEEE Int. Conf. Comput. Vision (ICCV)}
	(\bibinfo{year}{2015}) \bibinfo{pages}{2812 -- 2820}.
	\bibitem[{Shahid et~al.(2016)Shahid, Perraudin, Kalofolias, Puy, and
		Vandergheynst}]{graphPCA2}
	\bibinfo{author}{N.~Shahid}, \bibinfo{author}{N.~Perraudin},
	\bibinfo{author}{V.~Kalofolias}, \bibinfo{author}{G.~Puy},
	\bibinfo{author}{P.~Vandergheynst},
	\newblock \bibinfo{title}{Fast robust {PCA} on graphs},
	\newblock \bibinfo{journal}{IEEE J. Select. Topics Signal Process.}
	\bibinfo{volume}{10} (\bibinfo{year}{2016}) \bibinfo{pages}{740--756}.
	\bibitem[{Tremblay et~al.(2016)Tremblay, Puy, Borgnat, Gribonval, and
		Vandergheynst}]{Tremblay}
	\bibinfo{author}{N.~Tremblay}, \bibinfo{author}{G.~Puy},
	\bibinfo{author}{P.~Borgnat}, \bibinfo{author}{R.~Gribonval},
	\bibinfo{author}{P.~Vandergheynst},
	\newblock \bibinfo{title}{{Accelerated spectral clustering using graph
			filtering of random signals}},
	\newblock \bibinfo{journal}{{ IEEE Int. Conf. Acoust. Speech Signal Process.}}
	(\bibinfo{year}{2016}).
	\bibitem[{Benzi et~al.(2016)Benzi, Kalofolias, Bresson, and
		Vandergheynst}]{benzisong}
	\bibinfo{author}{K.~Benzi}, \bibinfo{author}{V.~Kalofolias},
	\bibinfo{author}{X.~Bresson}, \bibinfo{author}{P.~Vandergheynst},
	\newblock \bibinfo{title}{Song recommendation with non-negative matrix
		factorization and graph total variation},
	\newblock \bibinfo{journal}{Proc. IEEE Int. Conf. Acoust. Speech Signal
		Process.}  (\bibinfo{year}{2016}) \bibinfo{pages}{2439--2443}.
	\bibitem[{Shuman et~al.(2016)Shuman, Faraji, and Vandergheynst}]{pyramidgraph}
	\bibinfo{author}{D.~I. Shuman}, \bibinfo{author}{M.~J. Faraji},
	\bibinfo{author}{P.~Vandergheynst},
	\newblock \bibinfo{title}{A multiscale pyramid transform for graph signals},
	\newblock \bibinfo{journal}{IEEE Trans. Signal Process.} \bibinfo{volume}{64}
	(\bibinfo{year}{2016}) \bibinfo{pages}{2119--2134}.
	\bibitem[{Segarra et~al.(2016)Segarra, Marques, Mateos, and
		Ribeiro}]{blinddeconvgraph1}
	\bibinfo{author}{S.~Segarra}, \bibinfo{author}{A.~G. Marques},
	\bibinfo{author}{G.~Mateos}, \bibinfo{author}{A.~Ribeiro},
	\newblock \bibinfo{title}{Blind identification of graph filters with multiple
		sparse inputs},
	\newblock \bibinfo{journal}{IEEE Int. Conf. Acoust. Speech Signal Process.
		(ICASSP)}  (\bibinfo{year}{2016}) \bibinfo{pages}{4099--4103}.
	\bibitem[{Segarra et~al.(2017)Segarra, Mateos, Marques, and
		Ribeiro}]{blinddeconvgraph2}
	\bibinfo{author}{S.~Segarra}, \bibinfo{author}{G.~Mateos},
	\bibinfo{author}{A.~G. Marques}, \bibinfo{author}{A.~Ribeiro},
	\newblock \bibinfo{title}{Blind identification of graph filters},
	\newblock \bibinfo{journal}{{IEEE} Trans. Signal Process.} \bibinfo{volume}{65}
	(\bibinfo{year}{2017}) \bibinfo{pages}{1146--1159}.
	\bibitem[{Chen et~al.(2015)Chen, Sandryhaila, Moura, and Kovacevic}]{chen1}
	\bibinfo{author}{S.~Chen}, \bibinfo{author}{A.~Sandryhaila},
	\bibinfo{author}{J.~M.~F. Moura}, \bibinfo{author}{J.~Kovacevic},
	\newblock \bibinfo{title}{Signal recovery on graphs: Variation minimization},
	\newblock \bibinfo{journal}{{IEEE} Trans. Signal Process.} \bibinfo{volume}{63}
	(\bibinfo{year}{2015}) \bibinfo{pages}{4609--4624}.
	\bibitem[{Sakiyama et~al.(2016)Sakiyama, Watanabe, and Tanaka}]{sakiyama}
	\bibinfo{author}{A.~Sakiyama}, \bibinfo{author}{K.~Watanabe},
	\bibinfo{author}{Y.~Tanaka},
	\newblock \bibinfo{title}{Spectral graph wavelets and filter banks with low
		approximation error},
	\newblock \bibinfo{journal}{IEEE Trans. Signal Inform. Process. over Networks}
	\bibinfo{volume}{2} (\bibinfo{year}{2016}) \bibinfo{pages}{230--245}.
	\bibitem[{Deutsch et~al.(2016)Deutsch, Ortega, and Medioni}]{deutsch}
	\bibinfo{author}{S.~Deutsch}, \bibinfo{author}{A.~Ortega},
	\bibinfo{author}{G.~Medioni},
	\newblock \bibinfo{title}{Manifold denoising based on spectral graph wavelets},
	\newblock \bibinfo{journal}{IEEE Int. Conf. Acoust. Speech Signal Process.
		(ICASSP)}  (\bibinfo{year}{2016}) \bibinfo{pages}{4673--4677}.
	\bibitem[{Onuki et~al.(2016)Onuki, Ono, Yamagishi, and Tanaka}]{onuki}
	\bibinfo{author}{M.~Onuki}, \bibinfo{author}{S.~Ono},
	\bibinfo{author}{M.~Yamagishi}, \bibinfo{author}{Y.~Tanaka},
	\newblock \bibinfo{title}{Graph signal denoising via trilateral filter on graph
		spectral domain},
	\newblock \bibinfo{journal}{IEEE Trans. Signal Inform. Process. Netw.}
	\bibinfo{volume}{2} (\bibinfo{year}{2016}) \bibinfo{pages}{137--148}.
	\bibitem[{Romero et~al.(2017)Romero, Ma, and Giannakis}]{kergraph1}
	\bibinfo{author}{D.~Romero}, \bibinfo{author}{M.~Ma}, \bibinfo{author}{G.~B.
		Giannakis},
	\newblock \bibinfo{title}{Kernel-based reconstruction of graph signals},
	\newblock \bibinfo{journal}{{IEEE} Trans. Signal Process.} \bibinfo{volume}{65}
	(\bibinfo{year}{2017}) \bibinfo{pages}{764--778}.
	\bibitem[{Mendes et~al.(2016)Mendes, Mendes, and Ara{\'u}jo}]{tomograms_gsp}
	\bibinfo{author}{R.~V. Mendes}, \bibinfo{author}{H.~C. Mendes},
	\bibinfo{author}{T.~Ara{\'u}jo},
	\newblock \bibinfo{title}{Signals on graphs: Transforms and tomograms},
	\newblock \bibinfo{journal}{Physica A: Stat. Mechanics
		Appl.} \bibinfo{volume}{450} (\bibinfo{year}{2016}) \bibinfo{pages}{1
		-- 17}.
	\bibitem[{Jestrovi{\'c} et~al.(2017{\natexlab{a}})Jestrovi{\'c}, Coyle, and
		Sejdi{\'c}}]{JESTROVIC2017483}
	\bibinfo{author}{I.~Jestrovi{\'c}}, \bibinfo{author}{J.~L. Coyle},
	\bibinfo{author}{E.~Sejdi{\'c}},
	\newblock \bibinfo{title}{A fast algorithm for vertex-frequency representations
		of signals on graphs},
	\newblock \bibinfo{journal}{Signal Process.} \bibinfo{volume}{131}
	(\bibinfo{year}{2017}{\natexlab{a}}) \bibinfo{pages}{483 -- 491}.
	\bibitem[{Jestrovi{\'c} et~al.(2017{\natexlab{b}})Jestrovi{\'c}, Coyle, and
		Sejdi{\'c}}]{JESTROVIC2017113}
	\bibinfo{author}{I.~Jestrovi{\'c}}, \bibinfo{author}{J.~L. Coyle},
	\bibinfo{author}{E.~Sejdi{\'c}},
	\newblock \bibinfo{title}{Differences in brain networks during consecutive
		swallows detected using an optimized vertex--frequency algorithm},
	\newblock \bibinfo{journal}{Neuroscience} \bibinfo{volume}{344}
	(\bibinfo{year}{2017}{\natexlab{b}}) \bibinfo{pages}{113 -- 123}.
	\bibitem[{Kotzagiannidis and Dragotti(2017{\natexlab{a}})}]{fri_gsp}
	\bibinfo{author}{M.~Kotzagiannidis}, \bibinfo{author}{P.~Dragotti},
	\newblock \bibinfo{title}{Sampling and reconstruction of sparse signals on
		circulant graphs -- an introduction to graph-fri},
	\newblock \bibinfo{journal}{Appl. Computat. Harmonic Anal.,} \bibinfo{volume}{In
		Press} (\bibinfo{year}{2017}{\natexlab{a}}).
	\bibitem[{Kotzagiannidis and Dragotti(2017{\natexlab{b}})}]{splines_gsp}
	\bibinfo{author}{M.~Kotzagiannidis}, \bibinfo{author}{P.~Dragotti},
	\newblock \bibinfo{title}{Splines and wavelets on circulant graphs},
	\newblock \bibinfo{journal}{Appl. Computat. Harmonic Anal.,}
	\bibinfo{volume}{In press} (\bibinfo{year}{2017}{\natexlab{b}}).
	\bibitem[{Venkitaraman et~al.(2015)Venkitaraman, Chatterjee, and
		Handel}]{ArunSampta15}
	\bibinfo{author}{A.~Venkitaraman}, \bibinfo{author}{S.~Chatterjee},
	\bibinfo{author}{P.~Handel},
	\newblock \bibinfo{title}{On {H}ilbert transform of signals on graphs},
	\newblock \bibinfo{journal}{Proc. Sampling Theory Appl.}
	(\bibinfo{year}{2015}).
	\bibitem[{Gabor(1946)}]{Gabor}
	\bibinfo{author}{D.~Gabor},
	\newblock \bibinfo{title}{Theory of communication},
	\newblock \bibinfo{journal}{J. Inst. Elec. Eng.} \bibinfo{volume}{93}
	(\bibinfo{year}{1946}).
	\bibitem[{Gray(2006)}]{gray_circulant}
	\bibinfo{author}{R.~M. Gray},
	\newblock \bibinfo{title}{Toeplitz and circulant matrices: A review}
	\bibinfo{volume}{2} (\bibinfo{year}{2006}) \bibinfo{pages}{155--239}.
	\bibitem[{Isufi et~al.(2017)Isufi, Loukas, Simonetto, and Leus}]{armagraph}
	\bibinfo{author}{E.~Isufi}, \bibinfo{author}{A.~Loukas},
	\bibinfo{author}{A.~Simonetto}, \bibinfo{author}{G.~Leus},
	\newblock \bibinfo{title}{Autoregressive moving average graph filtering},
	\newblock \bibinfo{journal}{{IEEE} Trans. Signal Process.} \bibinfo{volume}{65}
	(\bibinfo{year}{2017}) \bibinfo{pages}{274--288}.
	\bibitem[{Oppenheim and Schafer(2009)}]{Oppenheim}
	\bibinfo{author}{A.~V. Oppenheim}, \bibinfo{author}{R.~W. Schafer},
	\bibinfo{title}{Discrete-Time Signal Processing},
	\bibinfo{publisher}{Prentice Hall Press}, \bibinfo{address}{Upper Saddle
		River, NJ, USA}, \bibinfo{edition}{3rd} edition, \bibinfo{year}{2009}.
	\bibitem[{Gold et~al.(1969)Gold, Oppenheim, and Rader}]{Gold}
	\bibinfo{author}{B.~Gold}, \bibinfo{author}{A.~V. Oppenheim},
	\bibinfo{author}{C.~M. Rader}, \bibinfo{title}{Theory and implementation of
		the Discrete {H}ilbert transform}, \bibinfo{type}{Technical Report}, Symp.
	Comput. Process. Commun., Polytechnic Institute of Brooklyn,
	\bibinfo{year}{1969}.
	\bibitem[{Horn and Johnson(2012)}]{Horn}
	\bibinfo{author}{R.~A. Horn}, \bibinfo{author}{C.~R. Johnson},
	\bibinfo{title}{Matrix Analysis}, \bibinfo{publisher}{Cambridge University
		Press}, \bibinfo{address}{New York, NY, USA}, \bibinfo{edition}{2nd} edition,
	\bibinfo{year}{2012}.
	\bibitem[{Edelman et~al.(1994)Edelman, Kostlan, and Shub}]{Adelman1}
	\bibinfo{author}{A.~Edelman}, \bibinfo{author}{E.~Kostlan},
	\bibinfo{author}{M.~Shub},
	\newblock \bibinfo{title}{How many eigenvalues of a random matrix are real?},
	\newblock \bibinfo{journal}{J. Amer. Math. Soc.} \bibinfo{volume}{7}
	(\bibinfo{year}{1994}).
	\bibitem[{Tao et~al.(2010)Tao, Vu, and Krishnapur}]{Tao}
	\bibinfo{author}{T.~Tao}, \bibinfo{author}{V.~Vu},
	\bibinfo{author}{M.~Krishnapur},
	\newblock \bibinfo{title}{Random matrices: Universality of esds and the
		circular law},
	\newblock \bibinfo{journal}{Ann. Probab.} \bibinfo{volume}{38}
	(\bibinfo{year}{2010}) \bibinfo{pages}{2023--2065}.
	\bibitem[{Cohen and Havlin(2010)}]{ErdosRenyi}
	\bibinfo{author}{R.~Cohen}, \bibinfo{author}{S.~Havlin},
	\bibinfo{title}{Complex Networks: Structure, Robustness and Function.},
	\bibinfo{publisher}{Cambridge University Press}, \bibinfo{year}{2010}.
	\bibitem[{Venkitaraman and Seelamantula(2014)}]{ArunFrHT}
	\bibinfo{author}{A.~Venkitaraman}, \bibinfo{author}{C.~S. Seelamantula},
	\newblock \bibinfo{title}{Fractional {H}ilbert transform extensions and
		associated analytic signal construction},
	\newblock \bibinfo{journal}{Signal Process.} \bibinfo{volume}{94}
	(\bibinfo{year}{2014}) \bibinfo{pages}{359 -- 372}.
	\bibitem[{Boashash(1992)}]{Boashash}
	\bibinfo{author}{B.~Boashash},
	\newblock \bibinfo{title}{Estimating and interpreting the instantaneous
		frequency of a signal. {I}. {F}undamentals},
	\newblock \bibinfo{journal}{Proc. IEEE} \bibinfo{volume}{80}
	(\bibinfo{year}{1992}) \bibinfo{pages}{520--538}.
	\bibitem[{Vakman(1998)}]{Vakman_book}
	\bibinfo{author}{D.~Vakman}, \bibinfo{title}{Signals, Oscillations, and Waves:
		A Modern Approach}, \bibinfo{publisher}{Artech House Boston},
	\bibinfo{year}{1998}.
	\bibitem[{Maragos et~al.(1993)Maragos, Kaiser, and Quatieri}]{Maragos}
	\bibinfo{author}{P.~Maragos}, \bibinfo{author}{J.~Kaiser},
	\bibinfo{author}{T.~Quatieri},
	\newblock \bibinfo{title}{On amplitude and frequency demodulation using energy
		operators},
	\newblock \bibinfo{journal}{IEEE Trans. Signal Process.} \bibinfo{volume}{41}
	(\bibinfo{year}{1993}) \bibinfo{pages}{1532--1550}.
	\bibitem[{Stokes and Handel(1998)}]{Handel_Maragos}
	\bibinfo{author}{V.~P. Stokes}, \bibinfo{author}{P.~Handel},
	\newblock \bibinfo{title}{Comments on "{O}n amplitude and frequency
		demodulation using energy operators"},
	\newblock \bibinfo{journal}{{IEEE} Trans. Signal Process.} \bibinfo{volume}{46}
	(\bibinfo{year}{1998}) \bibinfo{pages}{506--507}.
	\bibitem[{Venkitaraman and Seelamantula(2014)}]{Arun_binaural}
	\bibinfo{author}{A.~Venkitaraman}, \bibinfo{author}{C.~Seelamantula},
	\newblock \bibinfo{title}{Binaural signal processing motivated generalized
		analytic signal construction and {AM-FM} demodulation},
	\newblock \bibinfo{journal}{IEEE/ACM Trans. Audio Speech Language Process.}
	\bibinfo{volume}{22} (\bibinfo{year}{2014}) \bibinfo{pages}{1023--1036}.
	\bibitem[{Schwartzkopf et~al.(2000)Schwartzkopf, Milner, Ghosh, Evans, and
		Bovik}]{Bovik}
	\bibinfo{author}{W.~Schwartzkopf}, \bibinfo{author}{T.~Milner},
	\bibinfo{author}{J.~Ghosh}, \bibinfo{author}{B.~Evans},
	\bibinfo{author}{A.~Bovik},
	\newblock \bibinfo{title}{Two-dimensional phase unwrapping using neural
		networks},
	\newblock \bibinfo{journal}{Proc. 4th IEEE Southwest Symp. Image Anal.
		Interpretation}  (\bibinfo{year}{2000}) \bibinfo{pages}{274--277}.
	\bibitem[{Tribolet(1977)}]{Tribolet}
	\bibinfo{author}{J.~Tribolet},
	\newblock \bibinfo{title}{A new phase unwrapping algorithm},
	\newblock \bibinfo{journal}{IEEE Trans. Acoust. Speech Signal Process.}
	\bibinfo{volume}{25} (\bibinfo{year}{1977}) \bibinfo{pages}{170--177}.
	\bibitem[{Venkitaraman and Seelamantula(2012)}]{Arun_amfm2}
	\bibinfo{author}{A.~Venkitaraman}, \bibinfo{author}{C.~Seelamantula},
	\newblock \bibinfo{title}{A technique to compute smooth amplitude, phase, and
		frequency modulations from the analytic signal},
	\newblock \bibinfo{journal}{IEEE Signal Process. Lett.} \bibinfo{volume}{19}
	(\bibinfo{year}{2012}) \bibinfo{pages}{623--626}.
	\bibitem[{Lohmann et~al.(1996)Lohmann, Mendlovic, and Zalevsky}]{Lohmann}
	\bibinfo{author}{A.~W. Lohmann}, \bibinfo{author}{D.~Mendlovic},
	\bibinfo{author}{Z.~Zalevsky},
	\newblock \bibinfo{title}{{F}ractional {H}ilbert transform},
	\newblock \bibinfo{journal}{Opt. Lett.} \bibinfo{volume}{21}
	(\bibinfo{year}{1996}) \bibinfo{pages}{281--283}.
	\bibitem[{Kominek et~al.(2003)Kominek, Black, and Ver}]{CMUArctic}
	\bibinfo{author}{J.~Kominek}, \bibinfo{author}{A.~W. Black},
	\bibinfo{author}{V.~Ver}, \bibinfo{title}{{CMU} {A}rctic Databases for Speech
		Synthesis}, \bibinfo{type}{Technical Report}, \bibinfo{year}{2003}.
	\bibitem[{}]{NTVDb}
	\bibinfo{title}{{N}orth {T}exas {V}owel
		{D}atabase-http://www.utdallas.edu/$\sim$~assmann/kidvow1/north\textunderscore
		texas\textunderscore vowel\textunderscore database.html}
	
\end{thebibliography}

\end{document}